\documentclass[12pt]{article}

\usepackage[T1]{fontenc}
\usepackage[active]{srcltx}
\usepackage{lmodern}
\usepackage{amsmath,amstext,amssymb,graphicx,graphics,color}
\usepackage{tikz}
\usepackage{pgfplots}
\usepackage{enumitem,mathtools,mathrsfs}
\usepackage{theorem}
\usepackage{cite}               
\usepackage{hyperref}           
\usepackage{bbm}                
\usepackage{fancybox}           
\usepackage{diagbox} 
\usepackage[section]{placeins}  
\usepackage[font=small,margin=1em]{caption}
\usepackage{subcaption}
\usepackage{mhchem}
\usepackage[margin=2.75cm]{geometry}
\theorembodyfont{\normalfont\slshape} 

\newtheorem{theorem}{Theorem}
\newtheorem{proposition}{Proposition}[section]
\newtheorem{corollary}[proposition]{Corollary}

\theoremstyle{break} 

\newenvironment{proof}%
{{\par\noindent \bf Proof. \nobreak}}%
{\nobreak \removelastskip \nobreak \hfill $\Box$ \medbreak}
{{\par\noindent \bf Proof \nobreak}}%
{\nobreak \removelastskip \nobreak \hfill $\Box$ \medbreak}
{{\par\noindent \bf Proof lemma. \nobreak}}%
{\nobreak \removelastskip \nobreak \bf End proof lemma. \medbreak}
\newenvironment{remark}{\par \medskip \noindent {\bf Remark. }\nobreak}{\par \medskip}
\usepackage{fancyhdr}
 \fancypagestyle{plain}{
  \fancyhead{}
  
}

\def\paragraph#1{{\bf #1\ }}
\newcommand{\expo}{\mathrm{e}}

\newcommand{\dd}{\mathrm{d}}

\newcommand{\cS}{\mathcal{S}}
\newcommand{\cH}{\mathcal{H}}

\usepackage[utf8]{inputenc}
\usepackage{unicode_character}

\title{Relocation without preference: A destination-agnostic Schelling-type metapopulation model}
\author{Fei Cao \footnotemark[1] \and Roberto Cortez \footnotemark[2]}

\begin{document}
\maketitle

\footnotetext[1]{Amherst College - Department of Mathematics, Amherst, MA 01002, USA}
\footnotetext[2]{Universidad Andrés Bello - Departamento de Matemáticas, Santiago, Chile}

\tableofcontents


\begin{abstract}
In this work, we propose and analyze a novel Schelling-type metapopulation model that examines how random relocations of families between neighborhoods can lead to segregation. The model consists of a large number of houses organized into $N$ neighborhoods with $L$ houses each, without any spatial structure. Houses can be occupied by either a blue or a red family, and families relocate---to an empty house selected uniformly at random---at a rate that depends only on the number of families of the other type within the same neighborhood. We study two mean-field regimes: the large $N$ limit with fixed $L$, and the large $L$ limit with fixed $N$. The associated mean-field systems of ODEs are derived, and their long-time behavior is investigated. As is often the case with Schelling-type models, we find a rich interplay between the model parameters and the social structure of the equilibrium distribution, which exhibits segregation in some parameter ranges. Our work demonstrates that segregation patterns can emerge even when the relocation mechanism is destination-agnostic.
\end{abstract}


\section{Introduction}\label{sec:1}

The field of sociophysics has emerged as a robust interdisciplinary framework that applies the rigorous tools of statistical mechanics and nonlinear dynamics to the study of human collective behavior \cite{castellano_statistical_2009,sen_sociophysics_2014,galam_gefen_shapir_1982}. At its core, sociophysics posits that social systems can be understood through the lens of interacting particle systems, where individuals/agents follow relatively simple behavioral rules that give rise to complex, sometimes counter-intuitive, macroscopic phenomena. Through the mathematical abstraction of social dynamics, we can isolate the governing principles (such as phase transitions, symmetry breaking, and self-organization) that underpin diverse collective phenomena, ranging from the dynamics of public opinion and wealth redistribution to the spatiotemporal emergence of urban residential patterns \cite{cao_bennati_2025,cao_iterative_2025,cao_uncovering_2022,cao_wealth_2026,naldi_mathematical_2010,yizhaq_mathematical_2004}.

Perhaps the most iconic contribution to this field is the residential segregation model introduced by Thomas Schelling \cite{schelling_1969_models,schelling_1971_dynamic,schelling_2006_micromotives}. Schelling's work demonstrated a profound and somewhat unsettling reality: high levels of macroscopic segregation can emerge from agents who possess only a mild preference for similar neighbors. In the classic Schelling framework, agents inhabit a lattice (or checkerboard) and evaluate their local ``happiness'' based on a threshold of similarity with their immediate neighbors. If an agent's neighborhood falls below this threshold, they become ``unhappy'' and move to the nearest available satisfactory location. The revolutionary insight here was the decoupling of individual intent from collective outcomes: even ``tolerant'' individuals who are comfortable being in a minority can inadvertently contribute to a total spatial separation of groups. In the language of physics, this represents a self-organizing process driven by local feedback loops.

While early iterations of Schelling-type models were confined to discrete lattices and local ``jumping'' rules, modern research has shifted toward \emph{metapopulation} models \cite{durrett_2014_exact,gargiulo_2017_emergent}. In these frameworks, the system is viewed as a collection of sub-populations (neighborhoods or patches) linked by migration. This shift allows for the application of mean-field theories and deterministic differential equations, providing a clearer path to rigorous analytical results.

A critical component of these models is the relocation mechanism. Traditional models often assume ``preference-based relocation'', where agents seek out specific destinations that maximize their utility \cite{shin_2014_theoretical}. However, the complexity of real-world social systems, ranging from economic constraints to simple information asymmetry, often leads to ``uniform'' or ``blind'' relocation. In such scenarios, the focus shifts from the choice of destination to the trigger for departure. When agents move without a specific preference for their destination, the spatial structure of the system is shaped not by where people want to be, but by the differential stability of the neighborhoods they leave behind. This leads to a fascinating interplay between local dynamics and global constraints, where segregation emerges as a ``frozen'' state of a system that has lost its diffusive mixing.

\subsection{Definition of the model}\label{subsec:1.1}

In this work, we consider a metapopulation version of Schelling's segregation model, which consists of $N \in \mathbb{N}_+$ neighborhoods or blocks, each of which has $L \in \mathbb{N}_+$ houses that can be either occupied by a family or empty. Families in our model are minimum (and indivisible) units and there are two types of families in total, which we call \emph{blue} and \emph{red}. To make our model analytically trackable (as least to certain extent) we disregard spatial (network) topology within the neighborhoods together with their physical location. As has been mentioned in a recent work \cite{durrett_2014_exact}, the aforementioned assumptions/simplifications are sometimes realistic from a modeling perspective, as one can think of neighborhoods as blocks in a metropolitan city where the sense of community extends beyond immediate neighbors to the broader neighborhood as a whole. Assume that there are $\rho_{-}\,N\,L$ blue families and $\rho_{+}\,N\,L$ red families, where $\rho_\pm \in (0,1)$ are fixed parameters such that $0 < \rho_{-}+\rho_+ < 1$, thus leaving $(1-\rho_{-}-\rho_+)\,N\,L$ empty houses in total. To describe our model in detail we shall fix some notations and terminologies first. We label the $N$ neighborhoods from $1$ to $N$ and denote $X^{[i],j}_t \in \{-1,0,1\}$ the state of the $j$-th house ($1\leq j\leq L$) in the $i$-th neighborhood ($1\leq i\leq N$) at time $t$, where $0$ means that the house is empty and $-1$ or $1$ represent that the house is occupied by a blue or red family, respectively. The dynamics of our Schelling-type model can be described according to the following rule: each family moves to some given empty house (possibly in a different neighborhood) at a random time generated by a Poisson clock (independent of everything else) with rate prescribed by
\begin{equation}\label{eq:descrition_of_rate}
\frac{1}{N\,L}\left[\beta + \lambda\cdot\textrm{Number of families of opposite type in the same neighborhood}\right],
\end{equation}
where $\beta \geq 0$ and $\lambda \geq 0$ are additional model parameters which we term the \emph{base jump rate} and the \emph{parameter of intolerance}, respectively. From a modeling point of view, the parameter $\beta$ attempts to capture the underlying psychological motive for families to move to (or explore) a new living environment (which leads to \emph{diffusion} from a mathematical perspective), and the parameter $\lambda$ quantifies the strength of interactions for families living within the same neighborhood (which generates \emph{nonlinearity} in its mean-field description). We remark here that the presence of the normalization factor $\frac{1}{N\,L}$ in \eqref{eq:descrition_of_rate} is necessary to ensure that each blue or red family relocates at a rate of order $1$. 
A more explicit formula for the rate described in \eqref{eq:descrition_of_rate} is possible as well, and it reads as
\begin{equation}\label{eq:formula_rate}
\frac{|X^{[i],j}|}{N\,L}\left[\beta + \lambda\,\sum\limits_{k=1}^L \frac{(1-X^{[i],j}\,X^{[i],k})|X^{[i],k}|}{2}\right].
\end{equation}
Indeed, when $X^{[i],j} = 0$ (meaning that the $j$-th house in the $i$-th neighborhood is vacant) the rate of relocation \eqref{eq:formula_rate} will be zero and we also have
\begin{equation*}
\frac{(1-X^{[i],j}\,X^{[i],k})|X^{[i],k}|}{2} =
\begin{cases}
0, &~\textrm{if $X^{[i],k} = 0$ or $X^{[i],k} = X^{[i],j}$},\\
1, &~\textrm{if $X^{[i],k} \neq 0$ and $X^{[i],k} \neq X^{[i],j}$},
\end{cases}
\end{equation*}
whenever $X^{[i],j} \neq 0$. Thus the two descriptions \eqref{eq:descrition_of_rate} and \eqref{eq:formula_rate} for the rate of relocation are actually equivalent. Figure \ref{fig:illustration_model} provides a brief illustration of our Schelling-type metapopulation model and highlights several of our primary results.
\begin{figure}[!t]
  \centering
  \includegraphics[width=.8\textwidth]{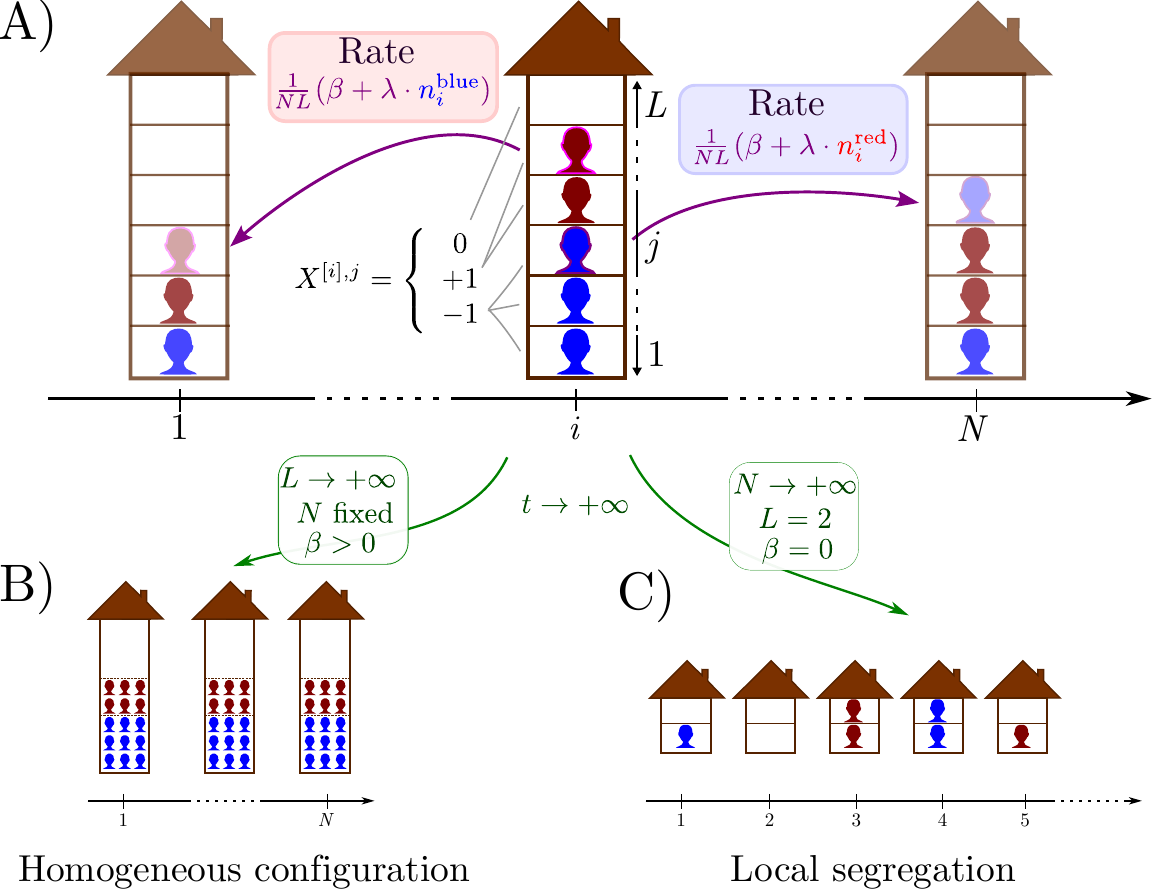}
  \caption{Illustration of the Schelling-type metapopulation model: families transition to vacant housing (potentially in different neighborhoods) at a rate defined by \eqref{eq:descrition_of_rate}. The number of blue and red families in the $i$-th neighborhood are denoted $n^{\text{blue}}_i$ and $n^{\text{red}}_i$, respectively. We analyze the system under two distinct asymptotic regimes: the limit $N \to \infty$ with $L$ fixed and the limit $L \to \infty$ with $N$ fixed. In Section \ref{sec:2}, we show that under the mean-field limit as $N \to \infty$ while keeping $L$ fixed (taking $L=2$ for the sake of simplicity), the system with $\beta = 0$ evolves toward \emph{local segregation} as $t \to \infty$. This is understood in the sense that, at equilibrium, blue and red families cannot co-reside in the same neighborhood. In Section \ref{sec:3}, we demonstrate that in the mean-field limit as $L \to \infty$ with $N$ fixed, which requires a suitable re-scaling of the rate in \eqref{eq:descrition_of_rate} by replacing $\lambda$ with $\lambda / L$, the system with $\beta > 0$ relaxes to a \emph{homogeneous configuration}, in the sense that the proportions of blue and red families are uniform across all neighborhoods.}
  \label{fig:illustration_model}
\end{figure}

We emphasize here that, even though our model setup draws important inspiration from the recent work \cite{durrett_2014_exact}, it differs from the typical Schelling-type dynamics in two crucial aspects:
\begin{enumerate}[label=(\roman*)]
\item Families (or agents) are not classified into ``happy'' or ``unhappy'' in our model, while a dominating majority of the literature on Schelling-type models introduce the notion of ``happiness'' to describe the status of the agent \cite{clark_2008_understanding,dall_2008_statistical}, measured in terms of certain utility function which depends on the number of agents of distinct type living in the neighborhood;
\item When a family moves, it will jump to an empty site/house uniformly at random according to our model description, whereas in most of the literature on Schelling-type models, the destination of a moving agent is biased toward (empty) sites which make the moving agent ``happier'' or at least less ``unhappy''. 
\end{enumerate}

\section{Mean-field limit as $N \to \infty$}\label{sec:2}
\setcounter{equation}{0}

Our model defines a stochastic agent-based dynamics involving a large number of (weakly) interacting/coupled families. To study its large-population asymptotic behavior, we resort to the framework of mean-field theories under suitable limiting procedures. Let $N^{b,r}_t$ be the total number of neighborhoods with $b$ blue families and $r$ red families at time $t$, where $(b,r) \in \mathcal{S} \coloneqq \{(b,r) \in \mathbb{N}^2 \mid b+r\leq L\}$; we say such a neighborhood is a \emph{$(b,r)$-neighborhood}. Since the spatial structure/topology of the set of $N$ neighborhoods is irrelevant to our discussion, the configuration/state of the system at any time $t\geq 0$ is fully characterized by $\{N^{b,r}_t\}_{(b,r)\in \mathcal{S}}$. Adopting a mean-field approach \cite{sznitman_topics_1991}, if we send $N \to \infty$ while keeping $L$ fixed, we expect that each $N^{b,r}_t/N$ will converge to a deterministic limit. Indeed, similar to the setup considered in \cite{durrett_2014_exact}, the framework provided in the work \cite{remenik_2009_limit} enables us to derive a mean-field system of ODEs for our model. Denote
\begin{equation}\label{eq:rate_movement}
\mathcal{R}(b_1,r_1;b_2,r_2) = \frac{1}{L}(L-b_2-r_2)\,\left[b_1\,(\beta + \lambda\,r_1) + r_1\,(\beta + \lambda\,b_1)\right]
\end{equation}
the normalized rate at which families in a given $(b_1,r_1)$-neighborhood move to a given $(b_2,r_2)$-neighborhood. Let $\mathcal{N}_i = (b_i,r_i)$ and $\mathcal{N}'_i = (b'_i,r'_i)$ for $i=1,2$. Then
\begin{equation}\label{eq:rate_transition}
\begin{aligned}
&\mathcal{J}(\mathcal{N}_1,\mathcal{N}_2;\mathcal{N}'_1,\mathcal{N}'_2) \\
&= \begin{cases}
\frac{b_1}{L}\,(L-b_2-r_2)\,(\beta + \lambda\,r_1), &~ \textrm{if $\mathcal{N}'_1 = (b_1-1,r_1)$, $\mathcal{N}'_2 = (b_1+1,r_1)$}; \\
\frac{r_1}{L}\,(L-b_2-r_2)\,(\beta + \lambda\,b_1), &~ \textrm{if $\mathcal{N}'_1 = (b_1,r_1-1)$, $\mathcal{N}'_2 = (b_1,r_1+1)$}; \\
0, &~ \textrm{otherwise}
\end{cases}
\end{aligned}
\end{equation}
yields the normalized rate at a which relocations from an $\mathcal{N}_1$-neighborhood to an $\mathcal{N}_2$-neighborhood transforms the pair $\mathcal{N}_1,\mathcal{N}_2$ into the pair $\mathcal{N}'_1,\mathcal{N}'_2$. Consequently, the mean-field ODE system can be recast into a similar form as in several previous works \cite{remenik_2009_limit,durrett_2014_exact}:

\begin{proposition}\label{prop:1}

For each $(b,r) \in \mathcal{S}$, the random fraction $N^{b,r}_t/N$ converges in probability, as $N \to \infty$, to the solution of the following coupled ODE system:
\begin{equation}\label{eq:mean_field_ODE}
	\begin{aligned}
		\frac{\dd}{\dd t} p_{b,r}(t) &=
		-p_{b,r}(t) \,\sum\limits_{\mathcal{N} \in \mathcal{S}} \left[\mathcal{R}(b,r;\mathcal{N})+\mathcal{R}(\mathcal{N};b,r)\right]\, p_{\mathcal{N}}(t) \\
		&~~+ \sum\limits_{\mathcal{N}_1,\mathcal{N}_2 \in \mathcal{S}} \left[\mathcal{J}(\mathcal{N}_1,\mathcal{N}_2;(b,r),\mathcal{N}') + \mathcal{J}(\mathcal{N}_1,\mathcal{N}_2;\mathcal{N}',(b,r))\right]\,p_{\mathcal{N}_1}(t) \, p_{\mathcal{N}_2}(t).
	\end{aligned}
\end{equation}

\end{proposition}


Notice that a summation over $\mathcal{N}'$ is unnecessary as the neighborhood configuration of $\mathcal{N}'$ is entirely determined by $(b,r)$, $\mathcal{N}_1$ and $\mathcal{N}_2$. The finite dimensional nonlinear (and deterministic) ODE system \eqref{eq:mean_field_ODE} is clearly of Boltzmann-type since one can interpret the first term (with a minus sign in front) as a consequence of relocations from a $(b,r)$-neighborhood to an $\mathcal{N}$-neighborhood or vice versa. Similarly, the second term takes into account migrations from an $\mathcal{N}_1$-neighborhood to an $\mathcal{N}_2$-neighborhood that may create a $(b,r)$-neighborhood at the departure or at the destination.

Consequently, in the second summation of \eqref{eq:mean_field_ODE}, the only non-zero terms involve neighborhoods with $b \pm 1$ blue families or $r \pm 1$ red families. We can thus write \eqref{eq:mean_field_ODE} more explicitly:
\begin{equation}
\label{eq:mean_field_ODE_explicit}
\begin{aligned}
	\frac{\dd}{\dd t} p_{b,r}(t)
	&= - p_{b,r}(t) \sum_{(b',r') \in \mathcal{S}} p_{b',r'}(t)
	\left[
	\mathcal{F}_{b,r}^{b', r'}
	+ \mathcal{G}_{b,r}^{b', r'}
	+ \mathcal{F}_{b',r'}^{b, r}
	+ \mathcal{G}_{b',r'}^{b, r}
	\right]
	\\
	& \quad + \sum_{(b',r') \in \mathcal{S}} p_{b',r'}(t)
	\left[
	p_{b+1,r}(t)\, \mathcal{F}_{b+1,r}^{b', r'}
	+p_{b,r+1}(t)\, \mathcal{G}_{b,r+1}^{b', r'}
	\right. \\
	& \qquad \qquad \qquad \qquad
	\left.
	+p_{b-1,r}(t)\, \mathcal{F}_{b',r'}^{b-1, r}
	+p_{b,r-1}(t)\, \mathcal{G}_{b',r'}^{b, r-1}
	\right],
\end{aligned}
\end{equation}
where $\mathcal{F}_{b,r}^{b', r'}$ and $\mathcal{G}_{b,r}^{b', r'}$ denote the normalized relocation rates of blue and red families, respectively, from a $(b,r)$-neighborhood to a $(b',r')$-neighborhood:
\begin{align*}
	\mathcal{F}_{b,r}^{b', r'}
	= \frac{b}{L}\,(\beta+\lambda r)\,(L-b'-r'),
	\qquad
	\mathcal{G}_{b,r}^{b', r'}
	= \frac{r}{L}\,(\beta+\lambda b)\,(L-b'-r'),
\end{align*}
with the convention that they are zero when either $(b,r)$ or $(b',r')$ does not belong to $\mathcal{S}$. In \eqref{eq:mean_field_ODE_explicit}, the meaning of each term is immediately apparent. For instance, the normalized rate at which $(b,r)$-neighborhoods are created at the destination as a consequence of blue families moving from a $(b',r')$-neighborhood is given by
\[p_{b',r'}(t)\, p_{b-1,r}(t)\, \mathcal{F}_{b',r'}^{b-1, r} = p_{b',r'}(t)\, p_{b-1,r}(t)\, \frac{b'}{L}\,(\beta + \lambda r')\, (L+1-b-r).\]

\subsection{Special case with $L=2$ houses per neighborhood}\label{subsec:2.1}

In order to demonstrate a more explicit version of the nonlinear ODE system \eqref{eq:mean_field_ODE_explicit}, we start with the special case where $L = 2$ (similar to the consideration illustrated in \cite{durrett_2014_exact}), meaning that each neighborhood consists of only $2$ houses in total.

In this setting, there are 6 types of neighborhoods, namely
\[
	\mathcal{S}
	= \{ (0,0) \, , \, (1,0) \, , \, (0,1) \, , \, (2,0) \, , \, (0,2) \, , \, (1,1) \}.
\]	
We summarize the relevant rates appearing in the limiting ODE system \eqref{eq:mean_field_ODE_explicit} in Table \ref{table:T1}.

\begin{table}[ht]
		\centering
		\begin{tabular}{cc}
			\begin{tabular}{c | c c c}
				$ (b,r) \, \big{\backslash} \, (b',r') $
				& $(0,0)$ & $(1,0)$ & $(0,1)$ \\
				\hline
				$(1,0)$ & $\beta$ & $\frac{\beta}{2}$ & $\frac{\beta}{2}$  \\
				$(0,1)$ & $0$ & $0$ & $0$  \\
				$(2,0)$ & $2\beta$ & $\beta$ & $\beta$  \\
				$(0,2)$ & $0$ & $0$ & $0$ \\
				$(1,1)$ & $\beta+\lambda$ & $\frac{\beta+\lambda}{2}$ & $\frac{\beta+\lambda}{2}$ \\
			\end{tabular}
			$~$ & $~$
			\begin{tabular}{c | c c c}
				$ (b,r) \, \big{\backslash} \, (b',r') $
				& $(0,0)$ & $(1,0)$ & $(0,1)$ \\
				\hline
				$(1,0)$ & $0$ & $0$ & $0$  \\
				$(0,1)$ & $\beta$ & $\frac{\beta}{2}$ & $\frac{\beta}{2}$ \\
				$(2,0)$ & $0$ & $0$ & $0$  \\
				$(0,2)$ & $2\beta$ & $\beta$ & $\beta$ \\
				$(1,1)$ & $\beta+\lambda$ & $\frac{\beta+\lambda}{2}$ & $\frac{\beta+\lambda}{2}$ \\
			\end{tabular}
		\end{tabular}
		\caption{Values of the rates $\mathcal{F}_{b,r}^{b',r'}$ (left) and $\mathcal{G}_{b,r}^{b',r'}$ (right) when $L = 2$. The row for $(0,0)$ and the columns for $(2,0)$, $(0,2)$ and $(1,1)$ are all zero.}
		\label{table:T1}
\end{table}

Consequently, after grouping and cancelling some terms, the nonlinear ODE system \eqref{eq:mean_field_ODE_explicit} boils down to:
\begin{equation}\label{eq:mean_field_L=2}
	\left\{
	\begin{aligned}
		p'_{0,0}
		&= -\left[2\,\beta\, p_{2,0}+2\,\beta\, p_{0,2}+2\,(\beta+\lambda)\,p_{1,1}\right]\,p_{0,0} + \tfrac{\beta}{2}\,p^2_{1,0} + \tfrac{\beta}{2}\,p^2_{0,1} +\beta \,p_{1,0}\,p_{0,1} \\
		p'_{1,0}
		&= -\left[\beta\,p_{1,0}+\beta\, p_{0,1}+\beta\, p_{0,2}+\tfrac{\beta+\lambda}{2}\,p_{1,1}\right]\,p_{1,0} \\
		&\quad {} + 4\,\beta\,p_{0,0}\,p_{2,0} + \beta\,p_{2,0}\,p_{0,1} + 2\,(\beta+\lambda)\,p_{1,1}\,p_{0,0}  + \tfrac{\beta+\lambda}{2}\,p_{0,1}\,p_{1,1} \\
		p'_{0,1}
		&= -\left[\beta\,p_{1,0}+\beta\, p_{0,1}+\beta\, p_{2,0}+\tfrac{\beta+\lambda}{2}\,p_{1,1}\right]\,p_{0,1} \\
		&\quad {} + 4\,\beta\,p_{0,0}\,p_{0,2} + \beta\,p_{0,2}\,p_{1,0} + 2\,(\beta+\lambda)\,p_{1,1}\,p_{0,0} + \tfrac{\beta+\lambda}{2}\,p_{1,0}\,p_{1,1} \\
		p'_{2,0}
		&= -\left[2\,\beta\,p_{0,0}+\beta\,p_{0,1}\right]\,p_{2,0} + \tfrac{\beta}{2}\,p^2_{1,0} + \tfrac{\beta+\lambda}{2}\,p_{1,0}\,p_{1,1} \\
		p'_{0,2}
		&= -\left[2\,\beta\,p_{0,0}+\beta\,p_{1,0}\right]\,p_{0,2} + \tfrac{\beta}{2}\,p^2_{0,1} + \tfrac{\beta+\lambda}{2}\,p_{0,1}\,p_{1,1} \\
		p'_{1,1}
		&= -\left[2\,(\beta+\lambda)\,p_{0,0}+\tfrac{\beta+\lambda}{2}\,p_{1,0}+\tfrac{\beta+\lambda}{2}\,p_{0,1}\right]\,p_{1,1} \\
		&\quad {} + \beta\,p_{1,0}\,p_{0,1} + \beta\,p_{2,0}\,p_{0,1} + \beta\,p_{0,2}\,p_{1,0}
	\end{aligned}
	\right.
\end{equation}
in which we have used primes to denote the derivative with respect to time $t$ in order to simplify the presentation. Now we can identify at least three (independent) conservation laws associated to the solution of the system \eqref{eq:mean_field_L=2}.

\begin{proposition}\label{prop:conserved_L=2}
Assume that ${\bf p} \coloneqq (p_{0,0},p_{1,0},p_{0,1},p_{2,0},p_{0,2},p_{1,1})^\intercal$ is a classical solution of the nonlinear ODE system \eqref{eq:mean_field_L=2} starting from an initial probability distribution ${\bf p}(t=0) = (p_{0,0}(0),p_{1,0}(0),p_{0,1}(0),p_{2,0}(0),p_{0,2}(0),p_{1,1}(0)) \in \mathbb{R}^6_+$ such that
\begin{equation}\label{eq:IC_L=2}
\begin{aligned}
&\frac 12\,\left(p_{1,1}(0)+p_{1,0}(0)+2\,p_{2,0}(0)\right) = \rho_{-}, \\
&\frac 12\,\left(p_{1,1}(0)+p_{0,1}(0)+2\,p_{0,2}(0)\right) = \rho_+, \\
&p_{0,0}(0) + p_{1,0}(0) + p_{0,1}(0) + p_{2,0}(0) + p_{0,2}(0) + p_{1,1}(0) = 1.
\end{aligned}
\end{equation}
Then for all $t\geq 0$ we have
\begin{equation}\label{eq:conservation_law}
\begin{aligned}
&\frac 12\,\left(p_{1,1}(t)+p_{1,0}(t)+2\,p_{2,0}(t)\right) = \rho_{-}, \\
&\frac 12\,\left(p_{1,1}(t)+p_{0,1}(0)+2\,p_{0,2}(t)\right) = \rho_+, \\
&p_{0,0}(t) + p_{1,0}(t) + p_{0,1}(t) + p_{2,0}(t) + p_{0,2}(t) + p_{1,1}(t) &= 1.
\end{aligned}
\end{equation}
\end{proposition}
The proof of Proposition \ref{prop:conserved_L=2} follows from a rather straightforward computation, hence it will be omitted. We emphasize here that the first two identities in \eqref{eq:conservation_law} imply the conservation of densities of blue and red families, respectively, while the last identity in \eqref{eq:conservation_law} means that the dynamical system \eqref{eq:mean_field_L=2} preserves the total probability mass. As an immediate consequence of the conservation laws \eqref{eq:conservation_law} we also deduce the following identity holding for all $t\geq 0$
\begin{equation}\label{eq:conserved_vacant_house}
\frac 12\,\left(2\,p_{0,0}(t)+p_{0,1}(0)+p_{1,0}(t)\right) = 1 - \rho_+ - \rho_{-},
\end{equation}
corresponding to conservation of the density of empty houses.

\subsubsection{Case $\beta=0$: Local segregation}

In the special case where $\beta = 0$ (that is, the effect of ``pure diffusion'' is turned off), the system \eqref{eq:mean_field_L=2} simplifies to
\begin{equation}\label{eq:mean_field_L=2,beta=0}
	\left\{\begin{aligned}
		p'_{0,0} &= -2\,\lambda\,p_{0,0}\,p_{1,1} \\
		p'_{1,0} &= \lambda\,\left[2\,p_{0,0}+\frac{p_{0,1}-p_{1,0}}{2}\right]\,p_{1,1} \\
		p'_{0,1} &= \lambda\,\left[2\,p_{0,0}+\frac{p_{1,0}-p_{0,1}}{2}\right]\,p_{1,1} \\
		p'_{2,0} &= \frac{\lambda}{2}\,p_{1,0}\,p_{1,1} \\
		p'_{0,2} &= \frac{\lambda}{2}\,p_{0,1}\,p_{1,1} \\
		p'_{1,1} &= -\lambda\,\left[2\,p_{0,0}+\frac{p_{1,0}+p_{0,1}}{2}\right]\,p_{1,1}.
	\end{aligned}\right.
\end{equation}
As we are dealing with the special case where $L = 2$, any neighborhood with both types of families must be a $(1,1)$-neighborhood. Therefore, we say that the mean-field Schelling-type dynamical system \eqref{eq:mean_field_L=2,beta=0} converges to a \emph{local segregation} state if $\lim_{t\to \infty} p_{1,1}(t) = 0$.
We can show the following result:

\begin{corollary}\label{coro:cv_to_local_segregation}
Under the settings of Proposition \ref{prop:conserved_L=2}, if $\beta = 0$ and $\lambda > 0$, then the solution ${\bf p}(t)$ of \eqref{eq:mean_field_L=2} converges to a local segregation state. More precisely, we have
\begin{equation}\label{eq:cv_of_p11}
p_{1,1}(t) \leq p_{1,1}(0)\,\expo^{-(1-\rho_+-\rho_{-})\lambda\,t}
\end{equation}
for all $t \geq 0$.
\end{corollary}

\begin{proof}
Thanks to the conservation of the density of empty houses \eqref{eq:conserved_vacant_house}, from \eqref{eq:mean_field_L=2,beta=0}, we deduce that
\begin{equation*}
p'_{1,1}
\leq -\frac{\lambda}{2}\,(2\,p_{0,0}+p_{0,1}+p_{1,0})\,p_{1,1}
= -\lambda\,(1-\rho_+-\rho_{-})\,p_{1,1},
\end{equation*}
from which one arrives at the announced bound \eqref{eq:cv_of_p11} via Gronwall's inequality.
\end{proof}

\begin{remark}
We remark here that the notion of local segregation introduced above makes intuitive sense, and it is a delicate task (if possible at all) to have a ``global'' notion of segregation, as our underlying agent-based model ignores the spatial/geographical interaction network between neighborhoods. Such simplification allows us to pass to a mean-field description without much difficulty, and we leave the systematic investigation of a spatial version/counterpart of our Schelling-type dynamics for future work.
\end{remark}

Although the significantly simplified dynamical system \eqref{eq:mean_field_L=2,beta=0} may look elementary, its equilibrium distribution ${\bf p}^* \coloneqq (p^*_{0,0},p^*_{1,0},p^*_{0,1},p^*_{2,0},p^*_{0,2},p^*_{1,1})^\intercal$ depends heavily on the initial datum ${\bf p}(0)$. In fact, any (time-independent) probability distribution ${\bf p}$ with $p_{1,1} = 0$ satisfying the set of compatibility conditions \eqref{eq:conservation_law} will be a candidate equilibrium distribution associated with the ODE system \eqref{eq:mean_field_L=2,beta=0}. On the other hand, it is not hard to argue that the equilibrium distribution ${\bf p}^*$ must satisfy $p^*_{1,1} = 0$. Therefore, at equilibrium, we have ${\bf p}^* \in \mathbb{R}^6_+$ together with
\begin{equation}\label{eq:equilibrium_linear_system_L=2,beta=0}
\left\{\begin{aligned}
&p^*_{1,0} + 2\,p^*_{2,0} = 2\,\rho_{-}, \\
&p^*_{0,1} + 2\,p^*_{0,2} = 2\,\rho_+, \\
&p^*_{0,0} + p^*_{1,0} + p^*_{0,1} + p^*_{2,0} + p^*_{0,2} = 1.
\end{aligned}\right.
\end{equation}
The linear system \eqref{eq:equilibrium_linear_system_L=2,beta=0} is obviously underdetermined, whence the distribution ${\bf p}^*$ to which the solution of \eqref{eq:mean_field_L=2,beta=0} converges will rely on the choices of initial datum ${\bf p}(0)$ and the parameters $\rho_\pm$. For instance, it is impossible to have $p^*_{0,0} = 0$ whenever $\rho_{-} + \rho_+ < 1/2$.

\subsubsection{Case $\lambda=0$: Convergence to the trinomial distribution}

On the other hand, if $\lambda = 0$ and $\beta > 0$ so that the interaction among families living in the same neighborhoods are turned off, the $6$-dimensional nonlinear ODE system \eqref{eq:mean_field_L=2} simplifies to
\begin{equation}\label{eq:mean_field_L=2,lambda=0}
\left\{\begin{aligned}
p'_{0,0} &= -2\,\beta\,p_{0,0}\,\left[p_{0,2}+p_{2,0}+p_{1,1}\right]+\frac{\beta}{2}\,\left[p^2_{1,0}+p^2_{0,1}\right]+\beta\,p_{1,0}\,p_{0,1} \\
p'_{1,0} &= -\beta\,p_{1,0}\,\left[1-p_{0,0}-p_{2,0}\right] + \beta\,\left[4\,p_{0,0}\,p_{2,0} + p_{0,1}\,p_{2,0} + 2\,p_{0,0}\,p_{1,1} + \frac{p_{1,0}+p_{0,1}}{2}\,p_{1,1}\right] \\
p'_{0,1} &= -\beta\,p_{0,1}\,\left[1-p_{0,0}-p_{0,2}\right] + \beta\,\left[4\,p_{0,0}\,p_{0,2} + p_{0,2}\,p_{1,0} + 2\,p_{0,0}\,p_{1,1} + \frac{p_{1,0}+p_{0,1}}{2}\,p_{1,1}\right] \\
p'_{2,0} &= -\beta\,p_{2,0}\,\left[2\,p_{0,0} + p_{0,1}\right] + \frac{\beta}{2}\,p^2_{1,0} + \frac{\beta}{2}\,p_{1,0}\,p_{1,1} \\
p'_{0,2} &= -\beta\,p_{0,2}\,\left[2\,p_{0,0} + p_{1,0}\right] + \frac{\beta}{2}\,p^2_{0,1} + \frac{\beta}{2}\,p_{0,1}\,p_{1,1} \\
p'_{1,1} &= -\beta\,p_{1,1}\,\left[2\,p_{0,0}+p_{0,1}+p_{1,0}\right] + \beta\,\left[p_{1,0}\,p_{0,1} + p_{2,0}\,p_{0,1} + p_{1,0}\,p_{0,2}+ \frac{p_{1,0}+p_{0,1}}{2}\,p_{1,1}\right].
\end{aligned}\right.
\end{equation}

It is straightforward to check that the distribution ${\bf p}^*$, defined by
\begin{equation}\label{eq:equili_L=2_lambda=0}
\begin{aligned}
	p^*_{0,0} &= (1-\rho_{-}-\rho_+)^2 \\
	p^*_{1,0} &= 2\,\rho_{-}\,(1-\rho_{-}-\rho_+) \\
	p^*_{0,1} &= 2\,\rho_+\,(1-\rho_{-}-\rho_+) \\
	p^*_{2,0} &= \rho^2_{-} \\
	p^*_{0,2} &= \rho^2_+ \\
	p^*_{1,1} &= 2\,\rho_{-}\,\rho_+
\end{aligned}
\end{equation}
is an equilibrium solution of the system \eqref{eq:mean_field_L=2,lambda=0}. In other words, ${\bf p}^*$ is a trinomial distribution given by
\begin{equation}\label{eq:equili_L=2_lambda=0_compact}
p^*_{b,r} = \binom{2}{b,r,2-b-r}\,\rho^b_{-}\,\rho^r_+\,(1-\rho_{-} - \rho_+)^{2-b-r}
\end{equation}
for each $(b,r) \in \cS$.~
Moreover, we can prove that the relative entropy from the solution ${\bf p}(t)$ to the trinomial equilibrium solution ${\bf p}^*$, defined by
\[
	\cH\left({\bf p}(t)\mid {\bf p}^*\right)
	\coloneqq \sum_{(b,r) \in \cS} p_{b,r}(t)\,\ln \frac{p_{b,r}(t)}{p^*_{b,r}},
\]
serves as a Lyapunov functional associated to the dynamical system \eqref{eq:mean_field_L=2,lambda=0}, which decreases monotonically with respect to time.

\begin{proposition}\label{prop:entropy_dissipation}
Under the settings of Proposition \ref{prop:conserved_L=2}, if $\lambda = 0$ and $\beta > 0$, then the solution ${\bf p}(t)$ of \eqref{eq:mean_field_L=2} satisfies
\begin{equation}\label{eq:dH/dt}
\begin{aligned}
&\frac{1}{\beta}\,\frac{\dd}{\dd t} \cH\left({\bf p}(t)\mid {\bf p}^*\right) \\
&= -2\,(p_{0,0}\,p_{2,0} - p^2_{1,0}/4)\,\ln\frac{4\,p_{0,0}\,p_{2,0}}{p^2_{1,0}} -2\,(p_{0,0}\,p_{0,2} - p^2_{0,1}/4)\,\ln\frac{4\,p_{0,0}\,p_{0,2}}{p^2_{0,1}}\\
& \qquad -2\,\left(p_{0,0}\,p_{1,1} - \frac{p_{0,1}\,p_{1,0}}{2}\right)\,\ln\frac{2\,p_{0,0}\,p_{1,1}}{p_{0,1}\,p_{1,0}} \\
& \qquad -\left(p_{0,1}\,p_{2,0} - \frac{p_{1,0}\,p_{1,1}}{2}\right)\,\ln\frac{2\,p_{0,1}\,p_{2,0}}{p_{1,0}\,p_{1,1}} -\left(p_{1,0}\,p_{0,2} - \frac{p_{0,1}\,p_{1,1}}{2}\right)\,\ln\frac{2\,p_{1,0}\,p_{0,2}}{p_{0,1}\,p_{1,1}} \\
& \leq 0
\end{aligned}
\end{equation}
for all $t\geq 0$, and the equality holds if and only if ${\bf p}$ coincides with ${\bf p}^*$. In particular, the solution ${\bf p}(t)$ converges to the trinomial distribution ${\bf p}^*$ as $t \to \infty$.
\end{proposition}

\begin{proof}
For the sake of notational simplicity and without loss of any generality, we work with the case where $\beta = 1$. The evolution of the relative entropy $\cH(t) \coloneqq \cH\left({\bf p}(t) \mid {\bf p}^*\right)$ is given by the sum of six terms as follows:
\begin{equation}\label{eq:dH}
\begin{aligned}
\cH'
&= p'_{0,0}\,(\ln p_{0,0} - \ln p^*_{0,0}) + p'_{1,1}\,(\ln p_{1,1} - \ln p^*_{1,1}) + p'_{1,0}\,(\ln p_{1,0} - \ln p^*_{1,0}) \\
&\quad {} + p'_{0,1}\,(\ln p_{0,1} - \ln p^*_{0,1}) + p'_{2,0}\,(\ln p_{2,0} - \ln p^*_{2,0}) + p'_{0,2}\,(\ln p_{0,2} - \ln p^*_{0,2}).
\end{aligned}
\end{equation}
Now we regroup the right-hand side in a specific manner and decompose the derivative $\cH'$ into a sum of five terms:
\begin{enumerate}[label=(\Alph*)]
\item Reaction A: $\ce{(1,0) + (1,0) <=> (0,0) + (2,0)}$, where we extract the $\left\{(0,0),(1,0),(2,0)\right\}$ interaction by grouping every term on the right-hand side of \eqref{eq:dH} that contains the product $p_{0,0}\,p_{2,0}$ or $p^2_{1,0}$:
    \begin{itemize}
    \item From $p'_{0,0}$: $-2\,p_{0,0}\,p_{2,0} + \frac 12\,p^2_{1,0}$;
    \item From $p'_{1,0}$: $4\,p_{0,0}\,p_{2,0} - p^2_{1,0}$;
    \item From $p'_{2,0}$: $-2\,p_{0,0}\,p_{2,0} + \frac 12\,p^2_{1,0}$;
    \end{itemize}
    Now we multiply these by their respective logarithms in the $\cH'$ sum to get
    \begin{align*}
    &\left(-2\,p_{0,0}\,p_{2,0} + \frac 12\,p^2_{1,0}\right)\left(\ln \frac{p_{0,0}}{p^*_{0,0}} + \ln \frac{p_{2,0}}{p^*_{2,0}} - 2\,\ln \frac{p_{1,0}}{p^*_{1,0}}\right) \\
    &\qquad = -2\,(p_{0,0}\,p_{2,0} - p^2_{1,0}/4)\,\ln\frac{4\,p_{0,0}\,p_{2,0}}{p^2_{1,0}} \leq 0.
    \end{align*}
\item Reaction B: $\ce{(0,1) + (0,1) <=> (0,0) + (0,2)}$, where we extract the $\left\{(0,0),(0,1),(0,2)\right\}$ interaction by grouping every term on the right-hand side of \eqref{eq:dH} that contains the product $p_{0,0}\,p_{0,2}$ or $p^2_{0,1}$:
    \begin{itemize}
    \item From $p'_{0,0}$: $-2\,p_{0,0}\,p_{0,2} + \frac 12\,p^2_{0,1}$;
    \item From $p'_{0,1}$: $4\,p_{0,0}\,p_{0,2} - p^2_{0,1}$;
    \item From $p'_{0,2}$: $-2\,p_{0,0}\,p_{0,2} + \frac 12\,p^2_{0,1}$;
    \end{itemize}
    Now we multiply these by their respective logarithms in the $\cH'$ sum to get
    \begin{align*}
    &\left(-2\,p_{0,0}\,p_{0,2} + \frac 12\,p^2_{0,1}\right)\left(\ln \frac{p_{0,0}}{p^*_{0,0}} + \ln \frac{p_{0,2}}{p^*_{0,2}} - 2\,\ln \frac{p_{0,1}}{p^*_{0,1}}\right) \\
    &\qquad = -2\,(p_{0,0}\,p_{0,2} - p^2_{0,1}/4)\,\ln\frac{4\,p_{0,0}\,p_{0,2}}{p^2_{0,1}} \leq 0.
    \end{align*}
\item Reaction C: $\ce{(0,0) + (1,1) <=> (1,0) + (0,1)}$, where we extract the $\left\{(0,0),(0,1),(1,0),(1,1)\right\}$ interaction by grouping every term on the right-hand side of \eqref{eq:dH} that contains the product $p_{0,0}\,p_{1,1}$ or $p_{0,1}\,p_{1,0}$:
    \begin{itemize}
    \item From $p'_{0,0}$: $-2\,p_{0,0}\,p_{1,1} + p_{1,0}\,p_{0,1}$;
    \item From $p'_{1,1}$: $-2\,p_{0,0}\,p_{1,1} + p_{1,0}\,p_{0,1}$;
    \item From $p'_{0,1}$: $-p_{1,0}\,p_{0,1} + 2\,p_{0,0}\,p_{1,1}$;
    \item From $p'_{1,0}$: $-p_{1,0}\,p_{0,1} + 2\,p_{0,0}\,p_{1,1}$;
    \end{itemize}
    Now we multiply these by their respective logarithms in the $\cH'$ sum to get
    \begin{align*}
    &\left(-2\,p_{0,0}\,p_{1,1} + p_{1,0}\,p_{0,1}\right)\left(\ln \frac{p_{0,0}}{p^*_{0,0}} + \ln \frac{p_{1,1}}{p^*_{1,1}} - \ln \frac{p_{0,1}}{p^*_{0,1}} - \ln \frac{p_{1,0}}{p^*_{1,0}}\right) \\
    &\qquad = -2\,\left(p_{0,0}\,p_{1,1} - \frac{p_{0,1}\,p_{1,0}}{2}\right)\,\ln\frac{2\,p_{0,0}\,p_{1,1}}{p_{1,0}\,p_{0,1}} \leq 0.
    \end{align*}
\item Reaction D: $\ce{(2,0) + (0,1) <=> (1,0) + (1,1)}$, where we extract the $\left\{(0,1),(1,0),(1,1),(2,0)\right\}$ interaction by grouping every term on the right-hand side of \eqref{eq:dH} that contains the product $p_{0,1}\,p_{2,0}$ or $p_{1,0}\,p_{1,1}$:
    \begin{itemize}
    \item From $p'_{2,0}$: $-p_{0,1}\,p_{2,0} + \frac 12\,p_{1,0}\,p_{1,1}$;
    \item From $p'_{1,1}$: $p_{0,1}\,p_{2,0} - \frac 12\, p_{1,0}\,p_{1,1}$;
    \item From $p'_{0,1}$: $-p_{0,1}\,p_{2,0} + \frac 12\,p_{1,0}\,p_{1,1}$;
    \item From $p'_{1,0}$: $p_{0,1}\,p_{2,0} - \frac 12\, p_{1,0}\,p_{1,1}$;
    \end{itemize}
    Now we multiply these by their respective logarithms in the $\cH'$ sum to get
    \begin{align*}
    &\left(-p_{0,1}\,p_{2,0} + \frac 12\, p_{1,0}\,p_{1,1}\right)\left(\ln \frac{p_{0,1}}{p^*_{0,1}} + \ln \frac{p_{2,0}}{p^*_{2,0}} - \ln \frac{p_{1,0}}{p^*_{1,0}} - \ln \frac{p_{1,1}}{p^*_{1,1}}\right) \\
    &\qquad = -\left(p_{0,1}\,p_{2,0} - \frac{p_{1,0}\,p_{1,1}}{2}\right)\,\ln\frac{2\,p_{0,1}\,p_{2,0}}{p_{1,0}\,p_{1,1}} \leq 0.
    \end{align*}
\item Reaction E: $\ce{(0,2) + (1,0) <=> (0,1) + (1,1)}$, where we extract the $\left\{(0,1),(1,0),(1,1),(0,2)\right\}$ interaction by grouping every term on the right-hand side of \eqref{eq:dH} that contains the product $p_{0,2}\,p_{1,0}$ or $p_{0,1}\,p_{1,1}$:
    \begin{itemize}
    \item From $p'_{0,2}$: $-p_{0,2}\,p_{1,0} + \frac 12\,p_{0,1}\,p_{1,1}$;
    \item From $p'_{1,1}$: $p_{0,2}\,p_{1,0} - \frac 12\, p_{0,1}\,p_{1,1}$;
    \item From $p'_{0,1}$: $p_{0,2}\,p_{1,0} - \frac 12\,p_{0,1}\,p_{1,1}$;
    \item From $p'_{1,0}$: $-p_{0,2}\,p_{1,0} + \frac 12\, p_{1,0}\,p_{1,1}$;
    \end{itemize}
    Now we multiply these by their respective logarithms in the $\cH'$ sum to get
    \begin{align*}
    &\left(-p_{0,2}\,p_{1,0} + \frac 12\, p_{0,1}\,p_{1,1}\right)\left(\ln \frac{p_{0,2}}{p^*_{0,2}} + \ln \frac{p_{1,0}}{p^*_{1,0}} - \ln \frac{p_{0,1}}{p^*_{0,1}} - \ln \frac{p_{1,1}}{p^*_{1,1}}\right) \\
    &\qquad = -\left(p_{0,2}\,p_{1,0} - \frac{p_{0,1}\,p_{1,1}}{2}\right)\,\ln\frac{2\,p_{0,2}\,p_{1,0}}{p_{0,1}\,p_{1,1}} \leq 0.
    \end{align*}
\end{enumerate}
Since each ``reaction'' contributes a non-positive term to $\cH'$, we finish the proof of \eqref{eq:dH/dt}. Moreover, $\cH' = 0$ only when the ratios inside the logarithms match the equilibrium ratios (which only happens at ${\bf p}^*$), thus $\cH' = 0$ if and only if ${\bf p} = {\bf p}^*$. Lastly, the desired (pointwise) convergence ${\bf p}(t) \to {\bf p}^*$ as $t\to \infty$ follows from a standard application of the entropy method in ordinary differential equations and dynamical systems (see for instance \cite{cao_derivation_2021,jungel_entropy_2016,perko_differential_2013}).
\end{proof}

\subsection{Numerical experiments}\label{subsec:2.2}


We perform a numerical investigation of the temporal evolution of the solution $\mathbf{p}(t)$ to the mean-field system \eqref{eq:mean_field_L=2} as it relaxes toward an equilibrium distribution. Simulation results are presented for three distinct parameter regimes: $\beta = 1$ and $\lambda = 0$ (Figure \ref{fig:F1}-left); $\beta = 0$ and $\lambda = 1$ (Figure \ref{fig:F2}-left); and $\beta = \lambda = 1$ (Figure \ref{fig:F2}-right). The system is initialized with $\mathbf{p}(0) = (0.05, 0.1, 0.2, 0.3, 0.15, 0.2)^\intercal$, so that $\rho_+ = 0.35$ and $\rho_{-} = 0.45$. To solve the ODE system \eqref{eq:mean_field_L=2} numerically, we employ a standard fourth-order Runge-Kutta scheme with a constant time step of $\Delta t = 0.01$.

\begin{figure}[!t]
  \begin{subfigure}{0.47\textwidth}
    \centering
    \includegraphics[width=\textwidth]{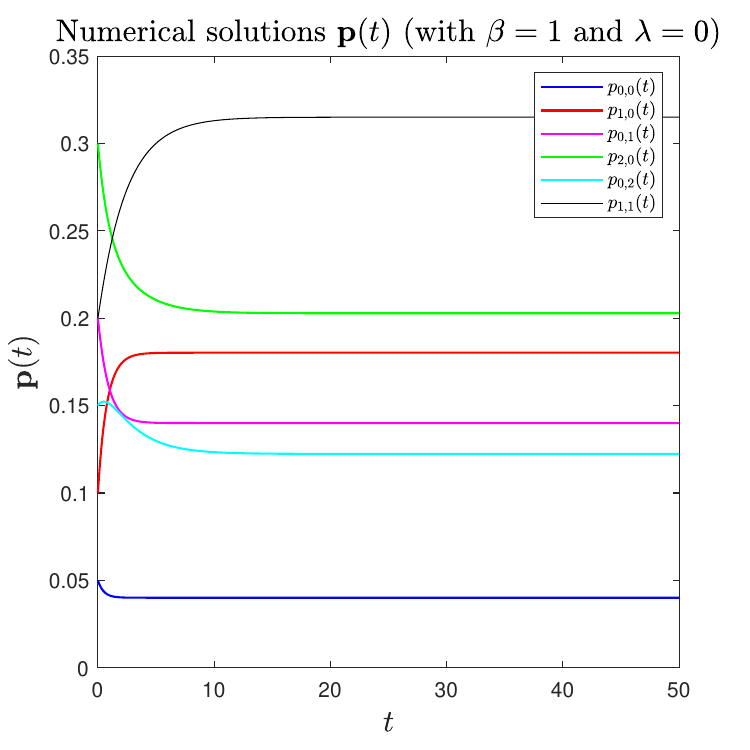}
  \end{subfigure}
  \hspace{0.1in}
  \begin{subfigure}{0.47\textwidth}
    \centering
    \includegraphics[width=\textwidth]{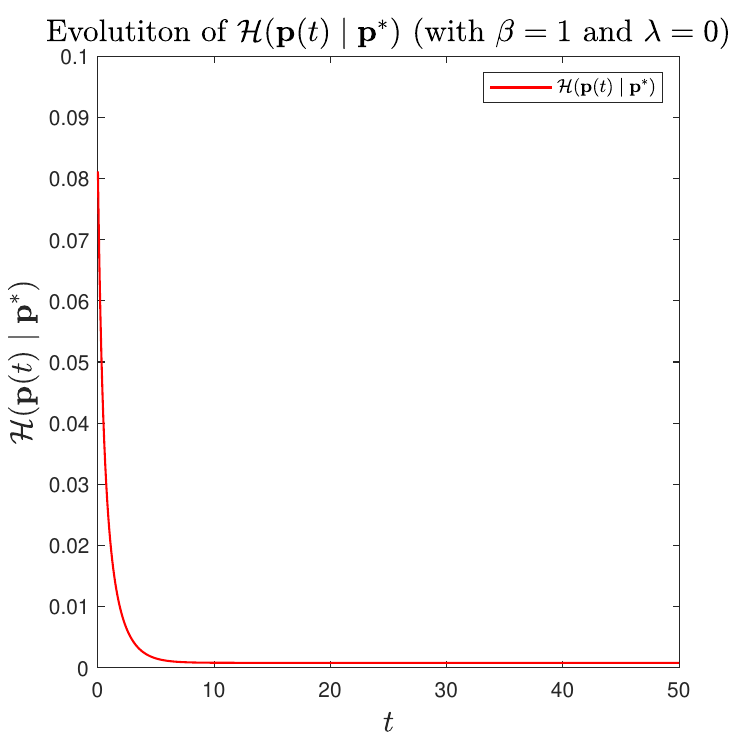}
  \end{subfigure}
  \caption{{\bf Left}: Evolution of the solution ${\bf p}(t)$ to the mean-field ODE system \eqref{eq:mean_field_L=2} with $\beta = 1$ and $\lambda = 0$. {\bf Right}: Evolution and dissipation of the relative entropy $\cH\left({\bf p}(t)\mid {\bf p}^*\right)$.}
  \label{fig:F1}  
\end{figure}

\begin{figure}[!t]
  \begin{subfigure}{0.47\textwidth}
    \centering
    \includegraphics[width=\textwidth]{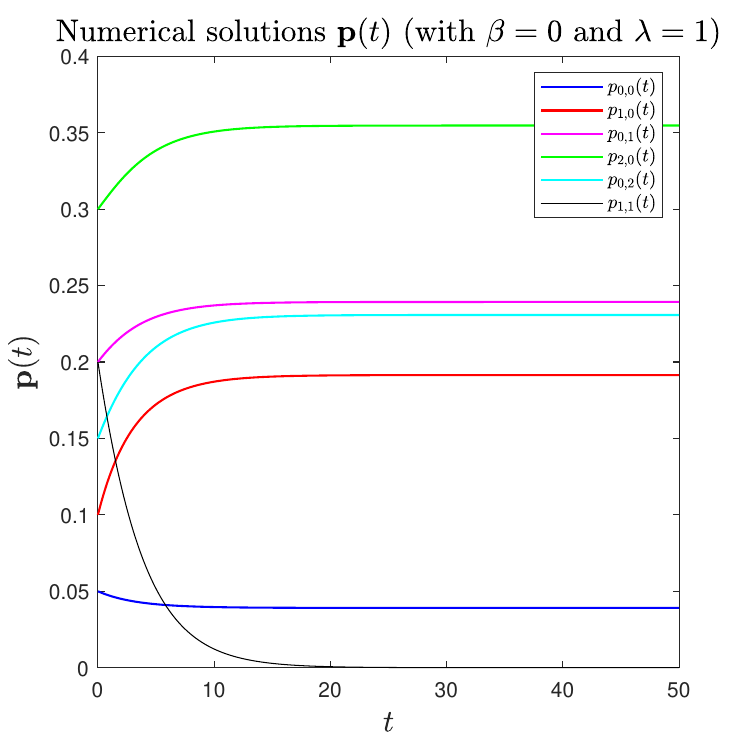}
  \end{subfigure}
  \hspace{0.1in}
  \begin{subfigure}{0.47\textwidth}
    \centering
    \includegraphics[width=\textwidth]{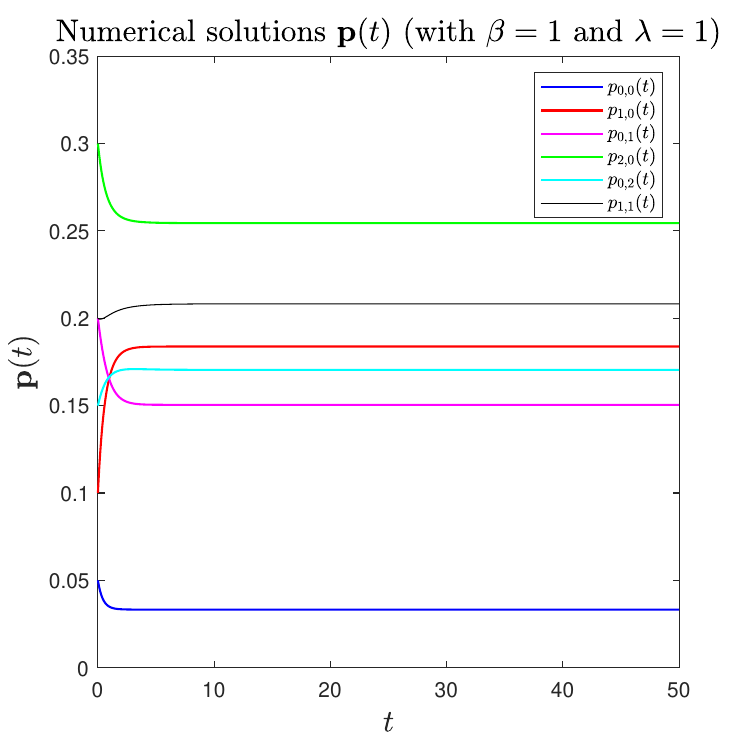}
  \end{subfigure}
  \caption{{\bf Left}: Evolution of the solution ${\bf p}(t)$ to the mean-field ODE system \eqref{eq:mean_field_L=2} with $\beta = 0$ and $\lambda = 1$. {\bf Right}: Evolution of the solution ${\bf p}(t)$ to the mean-field ODE system \eqref{eq:mean_field_L=2} with $\beta  = \lambda = 1$.}
  \label{fig:F2}  
\end{figure}

We see that in the absence of interaction-driven preferences (i.e., when $\lambda = 0$ and $\beta > 0$), the trajectories of the solution ${\bf p}(t)$ relax toward the trinomial distribution ${\bf p}^* = (0.04, 0.18, 0.14, 0.2025, 0.1225, 0.315)^\intercal$ prescribed by \eqref{eq:equili_L=2_lambda=0_compact}, which represents a perfectly mixed population. We also plot in Figure \ref{fig:F1}-right the evolution of the relative entropy $\cH\left({\bf p}(t)\mid {\bf p}^*\right)$ for the case $\beta = 1$ and $\lambda > 0$. As shown, the entropy decreases monotonically over time, providing numerical validation for the dissipation result established in Proposition \ref{prop:entropy_dissipation}. On the other hand, when diffusive mixing is suppressed (i.e., when $\beta = 0$ and $\lambda > 0$), we observe the emergence of (local) segregation, in the sense that $p_{1,1}(t) \to 0$ as $t \to \infty$, as illustrated in Figure \ref{fig:F2}-left.

\section{Mean-field limit as $L \to \infty$}\label{sec:3}
\setcounter{equation}{0}

In this section, we focus on an alternative limiting procedure: we let $L \to \infty$ while keeping $N$ frozen. Intuitively, each neighborhood $i \in \{1,\ldots,N\}$ can be viewed as a city or country inhabited by infinitely many families. In order to obtain a meaningful scaling limit as $L \to \infty$ with fixed $N$, the rate of moving prescribed in \eqref{eq:descrition_of_rate} or equivalently \eqref{eq:formula_rate} has to be renormalized appropriately. More precisely, we replace the jump rate  per family and empty house \eqref{eq:descrition_of_rate} with the following expression:
\begin{equation}\label{eq:descrition_of_rate_normalized}
\frac{1}{N\,L}\left[\beta + \lambda\cdot\textrm{Fraction of families of opposite type in the same neighborhood}\right].
\end{equation}
In other words, the rate \eqref{eq:formula_rate} is modified to
\begin{equation}\label{eq:formula_rate_normalized}
\frac{|X^{[i],j}|}{N\,L}\left[\beta + \frac{\lambda}{L}\,\sum\limits_{k=1}^L \frac{(1-X^{[i],j}\,X^{[i],k})|X^{[i],k}|}{2}\right].
\end{equation}
In the mean-field regime where $L \to \infty$ and $N$ is fixed, we denote $f_{i,+}(t)$ and $f_{i,-}(t)$ the fraction of red and blue families (respectively) living in the $i$-th neighborhood at time $t$, and call $f_{i,0} \coloneqq 1 - f_{i,+} - f_{i,-}$ the fraction of empty houses in the $i$-th neighborhood. Then, by following the mean-field framework established in \cite{cao_derivation_2021, cao_fractal_2024, cao_fractal_2026, cortez_quantitative_2016, cortez_uniform_2016, cortez_fontbona_2018}, a standard argument implies that the set of fractions $\left\{f_{i,\pm}\right\}_{1\leq i\leq N}$ satisfy the following coupled system of nonlinear ODEs:
\begin{equation}\label{eq:ODE_system}
\frac{\dd}{\dd t} f_{i,\pm} = f_{i,0}\,\sum_{j\neq i} (\beta + \lambda\,f_{j,\mp})\,f_{j,\pm} - f_{i,\pm}\,(\beta + \lambda\,f_{i,\mp})\,\sum_{j\neq i} f_{j,0}.
\end{equation}

\begin{remark}
We emphasize that the mean-field system \eqref{eq:ODE_system} is significantly more tractable than the system \eqref{eq:mean_field_ODE} discussed in the previous section. Indeed, the dimensionality of \eqref{eq:ODE_system} scales linearly with $N$, while the dimensionality of \eqref{eq:mean_field_ODE} suffers from combinatorial growth as $L$ increases.
\end{remark}

Next, we summarize some fundamental properties of the solutions to the coupled system \eqref{eq:ODE_system}, whose proof will be skipped as it involves only elementary computations.
\begin{proposition}\label{prop:conserved}
Assume that $\left\{f_{i,\pm}(t)\right\}_{1\leq i\leq N}$ is a classical solution of the nonlinear ODE system \eqref{eq:ODE_system} starting from a initial configuration $\left\{f_{i,\pm}(0)\right\}_{1\leq i\leq N}$ such that
\begin{equation}\label{eq:IC}
\frac{1}{N}\,\sum_{i=1}^N f_{i,\pm}(0) = \rho_{\pm}
\end{equation}
Then for all $t\geq 0$ it holds that
\begin{equation}\label{eq:conservation_law_ODE}
\frac{1}{N}\,\sum_{i=1}^N f_{i,\pm}(t) = \rho_{\pm}.
\end{equation}
Moreover, the homogeneous configuration $\{f^*_{i,\pm}\}_{1\leq i\leq N}$ defined by $f^*_{i,\pm} = \rho_\pm$ for all $1\leq i\leq N$, is an equilibrium solution to the system \eqref{eq:ODE_system}.
\end{proposition}


\subsection{Case $\lambda = 0$: Convergence to the homogeneous equilibrium}\label{subsec:3.1}

In the simplest case where $\lambda = 0$, if we denote $\rho_0 \coloneqq 1 - \rho_+ - \rho_{-} \in (0,1)$, then the nonlinear dynamical system \eqref{eq:ODE_system} reduces to
\begin{equation}\label{eq:ODE_lambda=0}
\begin{aligned}
f'_{i,\pm}
&= \beta\,f_{i,0}\,\sum_{j\neq i} f_{j,\pm} - \beta\,f_{i,\pm}\,\sum_{j\neq i} f_{j,0} \\
&= \beta\left[N\,\rho_\pm\,f_{i,0} - f_{i,0}\,f_{i,\pm} - N\,\rho_0\,f_{i,\pm} + f_{i,0}\,f_{i,\pm}\right] \\
&= N\,\beta\,\left(\rho_\pm\,f_{i,0} - \rho_0\,f_{i,\pm}\right),
\end{aligned}
\end{equation}
in which we have invoked the conservation laws \eqref{eq:conservation_law_ODE}. This is a linear system which is amenable to explicit solutions, and we have the following convergence guarantee:

\begin{theorem}\label{thm:exponential_conv}
Under the setting of Proposition \ref{prop:conserved}, if $\lambda = 0$ and $\beta > 0$, then for each $1\leq i\leq N$, we have $f_{i,\pm}(t) \to f^*_{i,\pm}$ as $t\to \infty$ exponentially with rate $N\,\beta\,\rho_0$.
\end{theorem}

\begin{proof}
For each $i \in \{1,\ldots, N\}$, we introduce $g_i = f_{i,+} - f_{i,-}$ and $h_i = f_{i,+} + f_{i,-}$. From \eqref{eq:ODE_lambda=0}, we obtain
\begin{equation*}
\left\{\begin{aligned}
&g'_i = N\,\beta\,\left((1-h_i)\,(\rho_+ - \rho_{-}) - g_i\,\rho_0\right), \\
&h'_i = N\,\beta\,\left((1-h_i)\,(\rho_+ + \rho_{-}) - h_i\,\rho_0\right) = N\,\beta\,\left(\rho_+ + \rho_{-} - h_i\right).
\end{aligned}\right.
\end{equation*}
Solving explicitly the ODE satisfied by $h_i$ leads us to
\[
h_i(t) = \rho_+ + \rho_{-} + \left[h_i(0) - (\rho_+ + \rho_{-})\right]\expo^{-N \beta t}.
\]
Inserting the expression for $h_i(t)$ into the ODE satisfied by $g_i$ and solving the resulting linear first-order equation, we obtain
\begin{align*}
g_i(t) &= \rho_+ - \rho_{-} + \left(g_i(0)-(\rho_+ - \rho_{-}) - \frac{(\rho_+ - \rho_{-})(h_i(0)-(\rho_+ + \rho_{-}))}{1-\rho_0}\right)\expo^{-N \beta \rho_0\, t} \\
&\quad + \frac{(\rho_+ - \rho_{-})(h_i(0)-(\rho_+ + \rho_{-}))}{1-\rho_0}\,\expo^{-N \beta\,t}.
\end{align*}
Finally, as $f_{i,+} = (h_i + g_i)/2$ and $f_{i,-} = (h_i - g_i)/2$, the advertised convergence guarantee follows.
\end{proof}

\subsection{Case $\beta = 0$: Existence of multiple equilibria}\label{subsec:3.2}

In the vanishing diffusion limit where $\beta = 0$, the nonlinear system \eqref{eq:ODE_system} reduces to
\begin{equation}\label{eq:ODE_beta=0}
	\begin{aligned}
		f'_{i,\pm} &= \lambda\,f_{i,0}\,\sum_{j\neq i} f_{j,+}\,f_{j,-} - \lambda\,f_{i,+}\,f_{i,-}\,\sum_{j\neq i} f_{j,0} \\
		&= \lambda\,f_{i,0}\,\sum_{j\neq i} f_{j,+}\,f_{j,-} - \lambda\,f_{i,+}\,f_{i,-}\,(N\,\rho_0 - f_{i,0}) \\
		&= \lambda\,\left(f_{i,0}\,\sum_{j=1}^N f_{j,+}\,f_{j,-} - N\,\rho_0\,f_{i,+}\,f_{i,-}\right).
	\end{aligned}
\end{equation}
In this case, the model admits other, non-homogeneous equilibrium solutions, beyond the homogeneous configuration $\{f^*_{i,\pm}\}$:


\begin{proposition}\label{prop:non-uniqueness}
Under the setting of Proposition \ref{prop:conserved} with $\beta = 0$ and $\lambda > 0$, there exist non-homogeneous equilibria for the nonlinear ODE system \eqref{eq:ODE_system}.
\end{proposition}

\begin{proof}
Indeed, we will produce multiple \emph{segregated} equilibria. We split the construction into the following scenarios depending on the relation among various model parameters:
\begin{enumerate}[label=(\roman*)]
\item If $\max\{\rho_+, \rho_{-}\} \geq \frac{N-1}{N}$, we assume without loss of generality that $\rho_{-} = \max\{\rho_+, \rho_{-}\}$. For each fixed $i_0 \in \{1,\ldots,N\}$, we define the configuration $\{\bar{f}_{i,\pm}\}_{1\leq i\leq N}$ by setting
\[
\bar{f}_{i_0,+} = N\,\rho_+,
\qquad
\bar{f}_{i_0,-} = N\,\rho_{-} - (N-1),
\]
and $\bar{f}_{i,+} = 0$, $\bar{f}_{i,-} = 1$ for all $i\neq i_0$. It is elementary to verify that these segregated configurations are indeed equilibrium solutions of the system \eqref{eq:ODE_beta=0}.

\item If $\max\{\rho_+, \rho_{-}\} < \frac{N-1}{N}$, we assume without loss of generality that $\rho_+ = \max\{\rho_+, \rho_{-}\}$. For each fixed $i_0 \in \{1,\ldots,N\}$, we define the configuration $\{\tilde{f}_{i,\pm}\}_{1\leq i\leq N}$ by setting
\[
\tilde{f}_{i_0,+} = 0,
\qquad
\tilde{f}_{i_0,-} = N\,\rho_{-},
\]
and $\tilde{f}_{i,+} = \frac{N}{N-1}\,\rho_+$, $\tilde{f}_{i,-} = 0$ for all $i\neq i_0$. It is elementary to verify that these segregated configurations are also equilibrium solutions of the system \eqref{eq:ODE_beta=0}.
\end{enumerate}
Putting together the discussions above allows us to conclude the proof.
\end{proof}

The previous result underscores the model's relevance: it demonstrates that the system admits non-homogeneous equilibria, effectively representing the phenomenon of spatial segregation in the absence of diffusive mixing. In addition, it is natural to speculate that there are uncountably many equilibria when $\beta = 0$. Indeed,  from \eqref{eq:ODE_beta=0} we see that $\frac{\dd}{\dd t} (f_{i,+} - f_{i,-}) = 0$, whence $f_{i,+}(t) - f_{i,-}(t) = f_{i,+}(0) - f_{i,-}(0)$ for all $t\geq 0$. That is, each neighborhood $i$ preserves its initial ``population bias'' $f_{i,+}(0) - f_{i,-}(0)$ as a local invariant. It is therefore natural to expect the system \eqref{eq:ODE_beta=0} to exhibit strong path dependency, where the specific initialization of these ``population bias'' determines which of the uncountably many equilibria the system ultimately converges to. In fact, the next result characterizes all possible equilibria:

\begin{theorem}
	\label{thm:rigid_mixing}
	Under the setting of Proposition \ref{prop:conserved} with $\beta = 0$ and $\lambda > 0$, any equilibrium configuration $\hat{f} = \{\hat{f}_{i,\pm}\}$ satisfies the following dichotomy:
	\begin{enumerate}[label=(\roman*)]
		\item either $\hat{f}$ is \textbf{fully segregated}, that is,
		\[
		\hat{f}_{i,+} = 0
		\quad \text{or} \quad
		\hat{f}_{i,-} = 0
		\quad \text{for all $1\leq i \leq N$},
		\]
		
		\item or $\hat{f}$ coincides with the \textbf{rigid mixing} configuration associated with the collection $\Delta_i := \frac{1}{2}(\hat{f}_{i,+} - \hat{f}_{i,+})$. Specifically, $\hat{f}$ is given by
		\begin{equation}
			\label{eq:rigid_mixing}
			\hat{f}_{i,\pm}
			= \sqrt{\alpha^2 + \alpha + \Delta_i^2} - \alpha \pm \Delta_i,
			\quad \text{for all $1\leq i \leq N$},
		\end{equation}
		where $\alpha>0$ is the unique solution to the algebraic equation
		\[
		\sum_{i=1}^N\sqrt{\alpha^2 + \alpha + \Delta_i^2}
		= N\alpha + \tfrac{1}{2} N (\rho_+ + \rho_-).
		\]
	\end{enumerate}
\end{theorem}

\begin{proof}
	From \eqref{eq:ODE_beta=0} with $0$ on the left-hand side, we see that $\hat{f}_{i,\pm}$ must satisfy the following algebraic relation:
	\begin{equation}
		\label{eq:equilibrium_algebraic_condition}
		\hat{f}_{i,0}\,\sum_{j=1}^N \hat{f}_{j,+}\,\hat{f}_{j,-}
		= \hat{f}_{i,+}\,\hat{f}_{i,-}\,N\,\rho_0,
		\quad \forall 1 \leq i \leq N.
	\end{equation}
	Assume $\hat{f}$ is not fully segregated, thus $\sum_j \hat{f}_{j,+ } \hat{f}_{j,-} > 0$. Let $m_i = \frac{1}{2} (\hat{f}_{i,+} + \hat{f}_{i,-})$, thus $\hat{f}_{i,\pm} = m_i \pm \Delta_i$. For a neighborhood $i$ such that $\hat{f}_{i,0} = 0$, \eqref{eq:equilibrium_algebraic_condition} implies that either $\hat{f}_{i,+} = 1$ or $\hat{f}_{i,-} = 1$, then $\Delta_i = 1/2$ and \eqref{eq:rigid_mixing} is satisfied for any $\alpha$. On the other hand, when $\hat{f}_{i,0} > 0$, from \eqref{eq:equilibrium_algebraic_condition}, we see that
	\[
	\frac{(m_i+\Delta_i)(m_i-\Delta_i)}{1-2 m_i}
	= \alpha
	\coloneqq \frac{1}{N \rho_0} \sum_{j=1}^N \hat{f}_{j,+} \hat{f}_{j,-}.
	\]
	Solving for $m_i$ gives
	\[
	m_i = \sqrt{\alpha^2 + \alpha + \Delta_i^2} -\alpha
	\]
	which yields the desired formula \eqref{eq:rigid_mixing}. It remains to characterize $\alpha$. Clearly $\sum_i m_i = \frac{1}{2}N (\rho_+ + \rho_-)$, thus $\alpha$ must satisfy
	\[
	0
	= -N\alpha - \tfrac{1}{2} N (\rho_+ + \rho_-) + \sum_{i=1}^N\sqrt{\alpha^2 + \alpha + \Delta_i^2}
	=: Q(\alpha).
	\]
	It is easy to check that $Q(0) < 0$, $\lim_{\alpha\to\infty} Q(\alpha) = \frac{1}{2} N \rho_0 > 0$ and that $Q$ is increasing. Thus, there exist a unique $\alpha>0$ such that $Q(\alpha)=0$. This concludes the proof.
\end{proof}

\begin{remark}
	The name ``rigid mixing'' for the configuration given by \eqref{eq:rigid_mixing} refers to the fact that this distribution attempts to mix the population as much as possible, under the constraint that every neighborhood has a prescribed population bias.
\end{remark}

\subsection{Case $\beta > 0$, $\rho_{+}=\rho_{-}$: Convergence to a unique equilibrium}\label{subsec:3.3}

As established in the preceding subsection, the uniqueness of equilibrium solutions is lost in the vanishing diffusion limit where $\beta = 0$. However, uniqueness of equilibrium can be restored for any $\beta > 0$, subject to the additional constraint that $\rho_+ = \rho_{-}$.

\begin{proposition}\label{prop:uniqueness_equal_rho}
Under the setting of Proposition \ref{prop:conserved}, if $\beta > 0$ and $\rho_+ = \rho_{-} = \rho \in (0, 1/2)$, then the homogeneous equilibrium $\{f^*_{i,\pm}\}$ is the unique equilibrium solution of the nonlinear system \eqref{eq:ODE_system}.
\end{proposition}

\begin{proof}
It suffices to show that whenever $\{F_{i,\pm}\}$ is another equilibrium solution, it must coincide with $\{f^*_{i,\pm}\}$. To start with, we claim that $F_{i,+} = F_{i,-}$ for each $1\leq i\leq N$. Indeed, if we denote $g_i = f_{i,+} - f_{i,-}$, we deduce from \eqref{eq:ODE_system} that
\begin{equation*}
\begin{aligned}
g'_i &= \beta\,f_{i,0}\,\sum_{j\neq i} g_j - \beta\,g_i\,\sum_{j\neq i} f_{j,0} \\
&= \beta\,f_{i,0}\,\sum_{j\neq i} g_j - \beta\,g_i\,N\,\rho_0 + \beta\,g_i\,f_{i,0} \\
&= N\,\beta\,\left(f_{i,0}\,(\rho_+ - \rho_{-}) - g_i\,\rho_0\right).
\end{aligned}
\end{equation*}
Under our additional assumption that $\rho_+ = \rho_{-}$, we obtain $g'_i = -N\,\beta\,\rho_0\,g_i$, whence $\lim_{t\to \infty} (f_{i,+}(t) - f_{i,-}(t)) = 0$. Consequently, any possible equilibrium configuration $\left\{F_{i,\pm}\right\}$ must satisfy the compatibility condition that $F_{i,+} - F_{i,-} = 0$ for all $1\leq i\leq N$. We denote $a_i = F_{i,+} = F_{i,-} \in (0,1/2)$ and introduce $f(x) = x\,(\beta + \lambda\,x)$ as well as $g(x) = 1-2\,x$, defined for $x \in (0,1/2)$. The stationarity of $\left\{F_{i,\pm}\right\}$ implies that $g(a_i)\,\sum_{j\neq i} f(a_j) = f(a_i)\,\sum_{j\neq i} g(a_j)$ for all $1\leq i\leq N$, from which we deduce that
\begin{equation}\label{eq:ratio}
\frac{f(a_i)}{g(a_i)} = \frac{\sum_{j=1}^N f(a_j)}{\sum_{j=1}^N g(a_j)}.
\end{equation}
The ratio on the right-hand side of the identity \eqref{eq:ratio} is clearly a constant independent of $i \in \{1,\ldots, N\}$. Let this constant be $K$. Thus for all $i \in \{1,\ldots, N\}$, we must have $R(a_i) = K$, where $R(x) = \frac{f(x)}{g(x)} = \frac{x\,(\beta + \lambda\,x)}{1-2\,x}$. To conclude that $a_i = \rho_+ = \rho_{-}$ for all $1\leq i\leq N$ (due to conservation law $\sum_{i=1}^N a_i = N\,\rho_+ = N\,\rho_{-}$), it suffices to notice the injectivity of the map $x \mapsto R(x)$ when $x \in (0,1/2)$.
\end{proof}

We now establish a convergence guarantee for the solution $\{f_{i,\pm}(t)\}$ of the nonlinear ODE system \eqref{eq:ODE_system} toward the homogeneous equilibrium $\{f^*_{i,\pm}\}$, assuming $\beta > 0$ and $\rho_+ = \rho_{-}$.

\begin{theorem}\label{thm:conv_homogeneous}
Under the setting of Proposition \ref{prop:conserved}, if $\beta > 0$ and $\rho_+ = \rho_{-} = \rho \in (0, 1/2)$, then it holds that $f_{i,\pm}(t) \xrightarrow{t\to \infty} f^*_{i,\pm}$ for each $1\leq i\leq N$.
\end{theorem}

\begin{proof}
For each $i \in \{1,\ldots, N\}$, we recall that $g_i \coloneqq f_{i,+} - f_{i,-}$ and $h_i \coloneqq f_{i,+} + f_{i,-}$. Thanks to the computations carried out in the proof of Proposition \ref{prop:uniqueness_equal_rho}, $g'_i = -N\,\beta\,\rho_0\,g_i$ and thus $g_i(t) = g_i(0)\,\expo^{-N\,\beta\,\rho_0\,t}$ for all $t\geq 0$. Now, a straightforward calculation also yields that
\begin{equation*}
\begin{aligned}
h'_i &= (1-h_i)\,\left(\beta\,\sum_{j=1}^N h_j + 2\,\lambda\,\sum_{j=1}^N f_{j,+}\,f_{j,-}\right) - N\,\rho_0\,\left(\beta\,h_i + 2\,\lambda\,f_{i,+}\,f_{i,-}\right) \\
&= N\,\beta\,(2\,\rho - h_i) + 2\,\lambda\,(1-h_i)\,\sum_{j=1}^N \frac{h^2_j - g^2_j}{4} - 2\,\lambda\,N\,\rho_0\,\frac{h^2_i - g^2_i}{4}.
\end{aligned}
\end{equation*}
Consequently, for $1\leq i,j \leq N$ we obtain
\begin{equation*}
\begin{aligned}
h'_i-h'_j &= -N\,\beta\,(h_i-h_j) - \frac{\lambda}{2}\,(h_i-h_j)\,\sum_{k=1}^N h^2_k - \frac{\lambda\,N\,\rho_0}{2}\,(h^2_i-h^2_j) \\
&\qquad + \frac{\lambda}{2}\,(h_i-h_j)\,\sum_{k=1}^N g^2_k + \frac{\lambda\,N\,\rho_0}{2}\,(g^2_i-g^2_j) \\
&= -(h_i - h_j)\left(N\,\beta + \frac{\lambda}{2}\,\sum_{k=1}^N (h^2_k - g^2_k)+ \frac{\lambda\,N\,\rho_0}{2}\,(h_i + h_j)\right) + \frac{\lambda\,N\,\rho_0}{2}\,(g^2_i-g^2_j) \\
&\leq -N\,\beta\,(h_i-h_j) + \frac{\lambda\,N\,\rho_0}{2}\,(g^2_i-g^2_j).
\end{aligned}
\end{equation*}
Inserting the explicit formulas for $g_i$ and $g_j$, we deduce from the aforementioned differential inequality that
\begin{equation*}
h_i(t) - h_j(t) \leq (h_i(0) - h_j(0))\,\expo^{-N\,\beta\,t} + \frac{(1-2\,\rho)\,\lambda\,(g^2_i(0)-g^2_j(0))}{2\,\beta\,(4\,\rho-1)}\left(\expo^{-(2-4\,\rho)\,N\,\beta\,t} - \expo^{-N\,\beta\,t}\right)
\end{equation*}
for $\rho \neq 1/4$, and
\begin{equation*}
h_i(t) - h_j(t) \leq \left(h_i(0) - h_j(0) + \frac{(1-2\,\rho)\,\lambda\,N\,(g^2_i(0)-g^2_j(0))}{2}\right)\expo^{-N\,\beta\,t}
\end{equation*}
for $\rho = 1/4$. Therefore, we obtain $\lim_{t \to \infty} (h_i(t) - h_j(t)) = 0$ for each $1\leq i,j \leq N$. As $\sum_{j=1}^N h_j(t) \equiv 2\,N\,\rho$ is conserved for all times, for each $1\leq i\leq N$ we have
\[\lim_{t \to \infty} (h_i(t) - 2\,\rho) = \lim_{t \to \infty} \left(h_i(t) - \frac{1}{N}\sum_{j=1}^N h_j(t)\right) = \frac{1}{N}\sum_{j=1}^N \lim_{t \to \infty} (h_i(t) - h_j(t)) = 0,\]
thus $\lim_{t\to\infty} h_i(t) = 2\rho$. Finally, we conclude that
\[\lim_{t \to \infty} f_{i,+}(t) = \lim_{t \to \infty} \frac{h_i(t) + g_i(t)}{2} = \rho \quad \textrm{and} \quad \lim_{t \to \infty} f_{i,+}(t) = \lim_{t \to \infty} \frac{h_i(t) - g_i(t)}{2} = \rho.\] This completes the proof.
\end{proof}

\subsection{Stability of the homogeneous equilibrium}\label{subsec:3.4}

For the asymmetric case where $\rho_+ \neq \rho_{-}$ and $\beta > 0$, a rigorous proof of uniqueness for the equilibrium solution to \eqref{eq:ODE_system} remains elusive. Nevertheless, numerical experiments presented in section \ref{subsec:3.5} suggest that the homogeneous states $\{f^*_{i,\pm}\}$ remains the unique global attractor whenever $\beta > 0$.

The primary objective of this subsection is to establish the stability and asymptotic stability of the homogeneous equilibrium $\{f^*_{i,\pm}\}$.

\begin{proposition}\label{prop:stability}
Under the settings of Proposition \ref{prop:conserved}, we have the following:
\begin{enumerate}[label=(\roman*)]
\item if $\beta = 0$ and $\lambda > 0$, then $\{f^*_{i,\pm}\}$ is Lyapunov stable;
\item if $\beta > 0$, then $\{f^*_{i,\pm}\}$ is asymptotically stable.
\end{enumerate}
\end{proposition}

\begin{proof}
We perform a standard linearization around the homogeneous equilibrium $\{f^*_{i,\pm}\}$. By setting $f_{i,\pm}(t) = \rho_{i,\pm} + \varepsilon\,r_{i,\pm}(t)$ with $|\varepsilon| \ll 1$, we obtain the following linearized system:
\begin{equation}
r'_{i,\pm} = -N\,(\beta + \lambda\,\rho_{\mp})\,(\rho_0 + \rho_{\pm})\,r_{i,\pm} + N\,\left(-\lambda\,\rho_0\,\rho_{\pm} - (\beta + \lambda\,\rho_{\mp})\,\rho_{\pm}\right)r_{i,\mp}
\end{equation}
for each $1\leq i\leq N$. Therefore, for each $1\leq i\leq N$, the pair $(r_{i,+}, r_{i,-})$ satisfies a two dimensional linear ODE system given by \begin{equation}\label{eq:linearization}
\frac{1}{N}\,\frac{\dd}{\dd t}\begin{pmatrix}
r_{i,+} \\ r_{i,-}
\end{pmatrix} = \begin{pmatrix}
-(\beta + \lambda\,\rho_{-})\,(\rho_0 + \rho_+) & -\lambda\,\rho_0\,\rho_+ - (\beta + \lambda\,\rho_{-})\,\rho_+  \\
-\lambda\,\rho_0\,\rho_{-} - (\beta + \lambda\,\rho_+)\,\rho_{-} & -(\beta + \lambda\,\rho_+)\,(\rho_0 + \rho_{-})
\end{pmatrix}\begin{pmatrix}
r_{i,+} \\ r_{i,-}
\end{pmatrix}
\end{equation}
The characteristic equation of the matrix in \eqref{eq:linearization} reads as
\begin{equation}\label{eq:character}
\begin{aligned}
0 &= z^2 + \left((1-\rho_{-})(\beta + \lambda\,\rho_{-}) + (1-\rho_+)(\beta + \lambda\,\rho_+)\right)z \\
&\quad + (1-\rho_{-})(1-\rho_+)(\beta + \lambda\,\rho_{-})(\beta + \lambda\,\rho_+) - (\beta + \lambda\,(1-\rho_+))\,\rho_+ - (\beta + \lambda\,(1-\rho_{-}))\,\rho_{-}
\end{aligned}
\end{equation}
Denoting the roots of \eqref{eq:character} as $z_1$ and $z_2$, it follows that
\[z_1 + z_2 = -\left((1-\rho_{-})(\beta + \lambda\,\rho_{-}) + (1-\rho_+)(\beta + \lambda\,\rho_+)\right) < 0.\]
If $\beta = 0$ and $\lambda > 0$, we obtain $z_1\,z_2 = 0$, whence $\{f^*_{i,\pm}\}$ is stable (in the sense of Lyapunov). When $\beta > 0$, we have
\begin{align*}
z_1\,z_2 &= (1-\rho_{-})(1-\rho_+)(\beta + \lambda\,\rho_{-})(\beta + \lambda\,\rho_+) - (\beta + \lambda\,(1-\rho_+))\,\rho_+ - (\beta + \lambda\,(1-\rho_{-}))\,\rho_{-} \\
&= (1-\rho_{-})(1-\rho_+)\left(\beta^2 + \beta\,\lambda\,(\rho_+ + \rho_{-})\right) - \rho_+\,\rho_{-}\left(\beta^2 + \beta\,\lambda\,(2 - \rho_+ -\rho_{-})\right) \\
&= \beta^2\,\rho_0 + \beta\,\lambda\,\left(\rho_0\,(1-\rho_0) - 2\,\rho_0\,\rho_+\,\rho_{-}\right) \\
&\geq \beta^2\,\rho_0 + \beta\,\lambda\,\frac{\rho_0\,(1-\rho_0)}{2} > 0.
\end{align*}
Consequently, we deduce that $\mathrm{Re}(z_1) < 0$ and $\mathrm{Re}(z_2) < 0$, hence $\{f^*_{i,\pm}\}$ is asymptotically stable.
\end{proof}

\subsection{Numerical experiments}\label{subsec:3.5}


We numerically investigate the temporal evolution of the solution $\{f_{i,\pm}(t)\}$ to the mean-field system \eqref{eq:ODE_system} for $N = 3$ as it converges toward equilibrium. Three distinct parameter regimes are explored: $\beta = 1, \lambda = 0$ (Figure \ref{fig:F3}-left); $\beta = 0, \lambda = 1$ (Figure \ref{fig:F3}-right); and $\beta = \lambda = 1$ (Figure \ref{fig:F4}). For the results shown in Figure \ref{fig:F3} and Figure \ref{fig:F4}-left, the system is initialized with $\mathbf{p}(0) = (0.55, 0.3, 0.2, 0.1, 0.15, 0.2)^\intercal$. In contrast, for Figure \ref{fig:F4}-right, we set $\mathbf{p}(0) = (0.4, 0.15, 0.3, 0.35, 0.2, 0.1)^\intercal$. Note that both initializations correspond to the population densities $\rho_+ = 0.3$ and $\rho_{-} = 0.2$. Numerical integration is performed using a standard fourth-order Runge-Kutta scheme with a constant time step $\Delta t = 0.01$.

\begin{figure}[!t]
  \begin{subfigure}{0.47\textwidth}
    \centering
    \includegraphics[width=\textwidth]{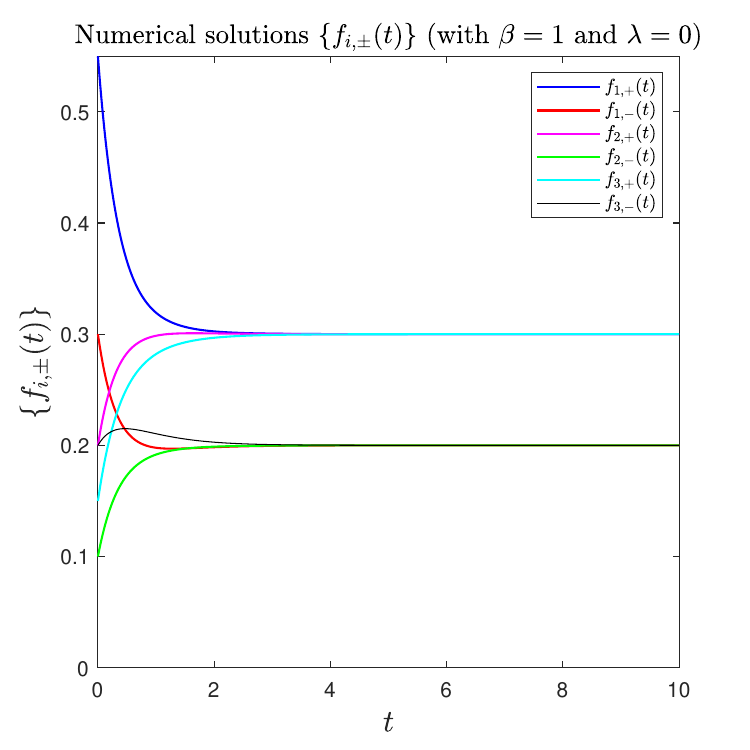}
  \end{subfigure}
  \hspace{0.1in}
  \begin{subfigure}{0.47\textwidth}
    \centering
    \includegraphics[width=\textwidth]{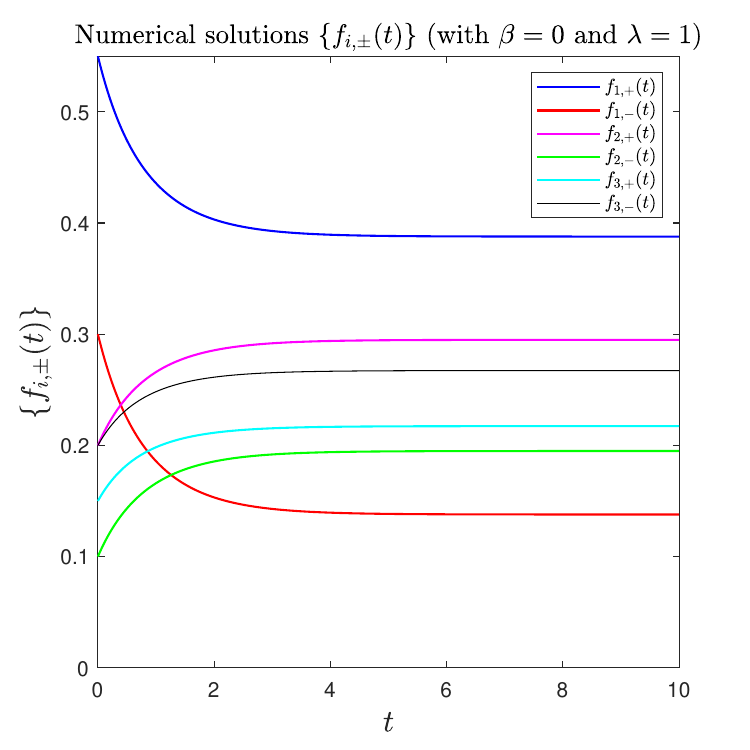}
  \end{subfigure}
  \caption{{\bf Left}: Evolution of the solution $\{f_{i,\pm}(t)\}$ to the mean-field ODE system \eqref{eq:ODE_system} for $N = 3$ with $\beta = 1$ and $\lambda = 0$. {\bf Right}: Evolution of the solution $\{f_{i,\pm}(t)\}$ to the mean-field ODE system \eqref{eq:ODE_system} for $N = 3$ with $\beta = 0$ and $\lambda = 1$.}
  \label{fig:F3}  
\end{figure}

\begin{figure}[!t]
  \begin{subfigure}{0.47\textwidth}
    \centering
    \includegraphics[width=\textwidth]{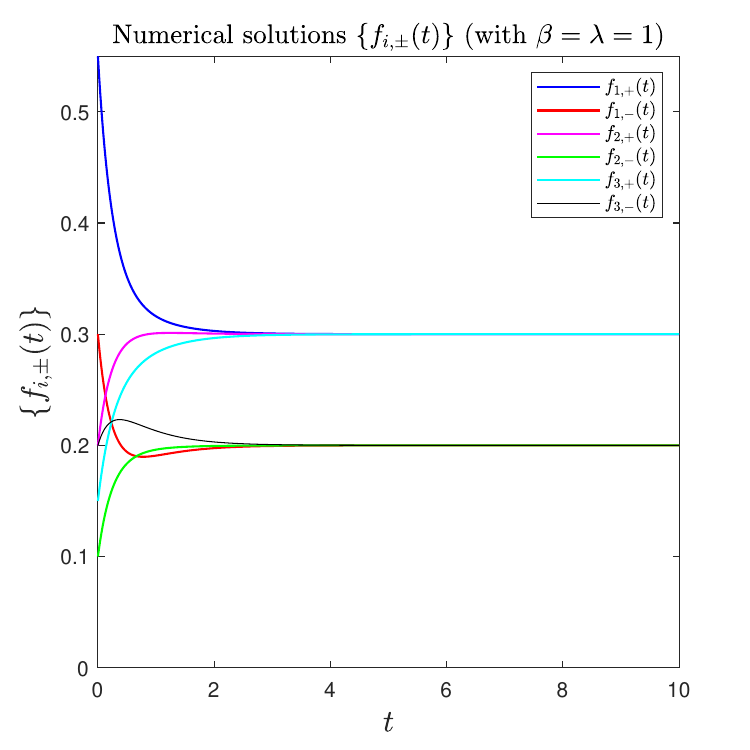}
  \end{subfigure}
  \hspace{0.1in}
  \begin{subfigure}{0.47\textwidth}
    \centering
    \includegraphics[width=\textwidth]{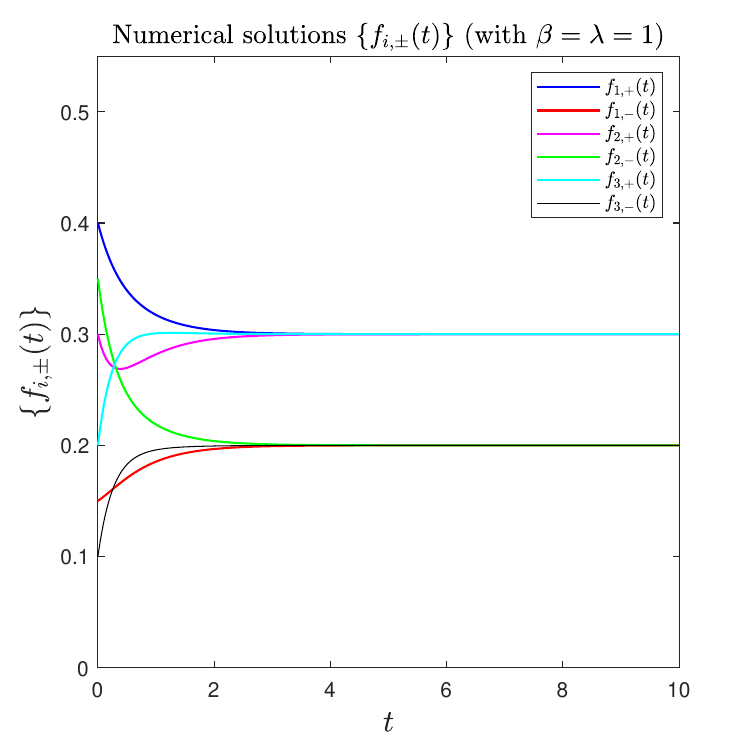}
  \end{subfigure}
  \caption{{\bf Left}: Evolution of the solution $\{f_{i,\pm}(t)\}$ to the mean-field ODE system \eqref{eq:ODE_system} for $N = 3$ with $\beta = \lambda = 1$, initialized with $\mathbf{p}(0) = (0.55, 0.3, 0.2, 0.1, 0.15, 0.2)^\intercal$. {\bf Right}: Evolution of the solution $\{f_{i,\pm}(t)\}$ to the mean-field ODE system \eqref{eq:ODE_system} for $N = 3$ with $\beta = \lambda = 1$, initialized with $\mathbf{p}(0) = (0.4, 0.15, 0.3, 0.35, 0.2, 0.1)^\intercal$.}
  \label{fig:F4}  
\end{figure}

Numerical results presented here indicate that without social interactions ($\lambda = 0$ and $\beta > 0$), the system relaxes to the homogeneous equilibrium $\{f^*_{i,\pm}\}$ (Figure \ref{fig:F3}-left). However, in the purely interactive regime ($\beta = 0$ and $\lambda > 0$), the system exhibits a ``frozen'' state where initial population biases are maintained indefinitely, as it converges to the rigid mixing distribution associated with these initial biases (described by \eqref{eq:rigid_mixing}), as illustrated in Figure \ref{fig:F3}-right. Most notably, Figure \ref{fig:F4} demonstrates that the introduction of any degree of diffusion ($\beta > 0$) restores the global attractivity of the homogeneous state, independent of the initialization or the presence of interaction terms.

\section{Conclusion}\label{sec:4}
\setcounter{equation}{0}

In this manuscript, we have introduced and investigated a Schelling-type metapopulation model that diverges from classical frameworks by prioritizing the mechanics of departure over the strategy of arrival. While our work is partially inspired by the metapopulation models of Durrett and Zhang \cite{durrett_2014_exact}, the primary novelty of our approach lies in the coupling of local environment-driven departures with a destination-agnostic relocation mechanism. By shifting the focus from utility-based optimization, where agents seek to maximize happiness at a specific destination, to a stochastic process where agents simply exit ``unstable'' environments, we demonstrate that spatial patterns can emerge even in the absence of intentional destination-seeking behavior.

Our analysis centered on two distinct mean-field regimes where the dynamical complexity differs significantly, and our analysis is focused on the more trackable mean-field regime where the number of houses $L$ (per neighborhood) tends to infinity while keeping the number of neighborhoods $N$ fixed. In the diffusive regime ($\beta > 0$), we established that the introduction of even a small degree of random movement acts as a powerful regularizing force. In the symmetric case ($\rho_+ = \rho_-$), we rigorously proved the uniqueness and global asymptotic stability of the homogeneous equilibrium. Conversely, the degenerate regime ($\beta = 0$) highlights a critical sociophysical insight: without a base-line level of mobility, initial social/population biases can become permanently locked into the spatial structure of a metapopulation.

The results presented here suggest that the ``push'' factor of environmental instability is a sufficient driver for collective dynamics, even when the ``pull'' factor of destination preference is absent. Moving forward, the study of finite-size effects and stochastic fluctuations in the underlying multi-agent system remains a promising avenue for future research. Ultimately, this work reinforces how simple local rules, even those lacking explicit utility/relocation optimization mechanism, can dictate the global evolution of complex social systems. \\

\noindent {\bf Acknowledgement~} We express our gratitude to Sebastien Motsch for generating Figure \ref{fig:illustration_model} that illustrates the Schelling-type metapopulation model proposed and studied in this manuscript. Fei Cao gratefully acknowledges support from an AMS-Simons Travel Grant, administered by the American Mathematical Society with funding from the Simons Foundation. Roberto Cortez acknowledges partial support from Fondecyt Grant 1242001.

\end{document}